\documentclass[sigconf,balance=false]{acmart}
\usepackage{popets}
\setcopyright{popets}
\copyrightyear{YYYY}

\acmYear{YYYY}
\acmVolume{YYYY}
\acmNumber{X}
\acmDOI{XXXXXXX.XXXXXXX}
\acmISBN{}
\acmConference{Proceedings on Privacy Enhancing Technologies}
\usepackage{hyperref}
\hypersetup{
	colorlinks=true,
	linkcolor=blue,
	citecolor=blue, %
	filecolor=blue, %
	urlcolor=blue %
}
\usepackage{amsmath,dsfont,amsthm}
\usepackage[nameinlink]{cleveref}
\usepackage{tikz}
\usetikzlibrary{calc,arrows.meta,decorations.pathreplacing}
\usepackage{xspace}
\usepackage{soul}  %
\usepackage{pgfplots}
\usepackage{caption}
\usepackage{subcaption}
\pgfplotsset{compat=1.16}  %
\usepackage{dsfont}
\usepackage{algorithm}
\usepackage[noend]{algorithmic}
\usepackage{wrapfig}

\usepackage[inline]{enumitem}
\setitemize[1]{leftmargin=4mm}
\setitemize[2]{leftmargin=3mm,label=$\triangleright$}
\setitemize[3]{leftmargin=3mm,label=$\square$}
\setenumerate[1]{leftmargin=4mm}

\newtheorem{theorem}{Theorem}[section]
\newtheorem{fact}[theorem]{Fact}
\newtheorem{remark}[theorem]{Remark}

\definecolor{Gred}{RGB}{219, 50, 54}
\definecolor{Ggreen}{RGB}{60, 186, 84}
\definecolor{Gblue}{RGB}{72, 133, 237}
\definecolor{Gyellow}{RGB}{247, 178, 16}
\definecolor{ToCgreen}{RGB}{0, 128, 0}
\definecolor{myGold}{RGB}{231,141,20}
\definecolor{myBlue}{rgb}{0.19,0.41,.65}
\definecolor{myPurple}{RGB}{175,0,124}

\def\ddefloop#1{\ifx\ddefloop#1\else\ddef{#1}\expandafter\ddefloop\fi}
\def\ddef#1{\expandafter\def\csname #1\endcsname{\ensuremath{\mathbb{#1}}}}
\ddefloop ABCDFGHIJKLMNOQRSTUVWXYZ\ddefloop  %
\def\ddef#1{\expandafter\def\csname c#1\endcsname{\ensuremath{\mathcal{#1}}}}
\ddefloop ABCDEFGHIJKLMNOPQRSTUVWXYZ\ddefloop  %
\def\ddef#1{\expandafter\def\csname b#1\endcsname{\ensuremath{\bm #1}}}
\ddefloop ABCDEFGHIJKLMNOPQRSTUVWXYZabcdeghijklnopqrstuvwxyz\ddefloop  %

\newcommand{\set}[1]{\{{#1}\}}

\newcommand{\Ad}{ad\xspace}
\newcommand{\mechanism}{\cM}
\newcommand{\indmechanism}{\mechanism^{\mathrm{ind}}}
\newcommand{\function}{\mathcal{F}}
\newcommand{\adversary}{\cA}
\newcommand{\adversaryDistribution}{\mathbb{A}}

\newcommand{\QuerySet}{\cQ}
\newcommand{\QuerySetU}{\cQ_{\cU}}
\newcommand{\ResponseSet}{\cR}
\newcommand{\StateSet}{\cS}

\newcommand{\DataSets}{\cD}

\newcommand{\dummyImpression}{x_\bot}
\newcommand{\dummySharedStorage}{y_\bot}
\newcommand{\RID}{\cI}
\newcommand{\aggRepKeys}{\cK}
\newcommand{\activeSources}{\cS}
\newcommand{\query}{q}

\newcommand{\mechState}{S}

\newcommand{\filter}{\phi}
\newcommand{\algcomment}[1]{\hfill {\em \textcolor{black!60}{[#1]}}}

\newcommand{\timeWindowSet}{\cT}
\newcommand{\deviceSet}{\cU}
\newcommand{\device}{u}

\makeatletter
\newcommand\xleftrightarrow[2][]{%
	\ext@arrow 9999{\longleftrightarrowfill@}{#1}{#2}}
\newcommand\longleftrightarrowfill@{%
	\arrowfill@\leftarrow\relbar\rightarrow}
\makeatother

\newcommand{\transcript}[2]{\mathrm{IT}\left(#1 : #2\right)}

\newcommand{\DLap}{\mathsf{DLap}}

\newcommand{\MSR}{\mechanism_{\mathrm{SR}}}
\newcommand{\MER}{\mechanism_{\mathrm{ER}}}

\newcommand{\eps}{\varepsilon}

\newcommand{\divergence}[1]{\mathsf{R}_{#1}}

\newcommand{\contribBudget}{\Lambda_1}%

\newcommand{\sparsityBudget}{\Lambda_0}%

\newcommand{\halt}{\text{\raisebox{-1pt}{\includegraphics[height=2.8mm]{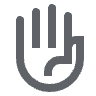}}}}

\newcommand{\thumbsup}{\text{\raisebox{-0.5pt}{\includegraphics[height=2.8mm]{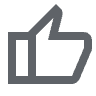}}}}
\newcommand{\thumbsdown}{\text{\raisebox{-1.5pt}{\includegraphics[height=2.8mm]{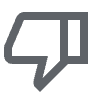}}}}

\newcommand{\AggregationService}{\textsf{AggregationService}\xspace}
\newcommand{\ARASRClient}{\textsf{ARA-SR-Client}\xspace}
\newcommand{\ARAEventClient}{\textsf{ARA-Event-Client}\xspace}
\newcommand{\PAAClient}{\textsf{PAA-SR-Client}\xspace}

\newcommand{\rroutput}{\mathcal{O}}
\newcommand{\rroutputvalid}{\mathcal{O}}
\newcommand{\randomizedoutput}{s^*}
\newcommand{\IRR}{\textsf{InteractiveRandomizedResponse}\xspace}

\newcommand{\sourceId}{{\texttt{\small srcId}}\xspace}
\newcommand{\destination}{{\texttt{\small dest}}\xspace}
\newcommand{\expDate}{{\texttt{\small expDate}}\xspace}
\newcommand{\sourceFilters}{{\texttt{\small srcFilt}}\xspace}
\newcommand{\sourceKey}{{\texttt{\small srcKey}}\xspace}

\newcommand{\triggerId}{{\texttt{\small trigId}}\xspace}
\newcommand{\triggerFilters}{{\texttt{\small trigFilt}}\xspace}
\newcommand{\triggerKey}{{\texttt{\small trigKey}}\xspace}
\newcommand{\triggerValue}{{\texttt{\small trigValue}}\xspace}
\newcommand{\triggerData}{{\texttt{\small trigData}}\xspace}
\newcommand{\triggerDataSet}{\mathsf{TD}}
\newcommand{\triggerSpecification}{{\texttt{\small trigSpec}}\xspace}
\newcommand{\maximalNumberOfReports}{{\texttt{\small maxRep}}\xspace}

\begin{document}

\title{
    On the Differential Privacy and Interactivity of Privacy Sandbox Reports
}

    \author{Badih Ghazi}
    \affiliation{%
    	\institution{Google}
    	\city{Mountain View}
    	\state{California}
    	\country{USA}
    }
    \email{badihghazi@gmail.com}
    \authornote{Alphabetical author order.}
    
    \author{Charlie Harrison}
    \affiliation{%
    	\institution{Google}
    	\state{Texas}
    	\country{USA}
    }
    \email{csharrison@google.com}
    \authornotemark[1]
    
    \author{Arpana Hosabettu}
    \affiliation{%
    	\institution{Google}
    	\city{Mountain View}
    	\state{California}
    	\country{USA}
    }
    \email{arpanah@google.com}
    \authornotemark[1]
    
    \author{Pritish Kamath}
    \affiliation{%
    	\institution{Google}
    	\city{Mountain View}
    	\state{California}
    	\country{USA}
    }
    \email{pritishk@google.com}
    \authornotemark[1]
    
    \author{Alexander Knop}
    \affiliation{%
    	\institution{Google}
    	\city{New York City}
    	\state{New York}
    	\country{USA}
    }
    \email{alexanderknop@google.com}
    \authornotemark[1]
    
    \author{Ravi Kumar}
    \affiliation{%
    	\institution{Google}
    	\city{Mountain View}
    	\state{California}
    	\country{USA}
    }
    \email{ravi.k53@gmail.com}
    \authornotemark[1]
    
    \author{Ethan Leeman}
    \affiliation{%
    	\institution{Google}
    	\city{Cambridge, Massachusetts}
    	\country{USA}
    }
    \email{ethanleeman@google.com}
    \authornotemark[1]
    
    \author{Pasin Manurangsi}
    \affiliation{%
    	\institution{Google}
    	\city{Bangkok}
    	\country{Thailand}
    }
    \email{pasin@google.com}
    \authornotemark[1]
    
    \author{Mariana Raykova}
    \affiliation{%
    	\institution{Google}
    	\city{New York City}
    	\state{New York}
    	\country{USA}
    }
    \email{marianar@google.com}
    \authornotemark[1]
    
    \author{Vikas Sahu}
    \affiliation{%
    	\institution{Google}
    	\city{Mountain View}
    	\state{California}
    	\country{USA}
    }
    \email{vikassahu@google.com}
    \authornotemark[1]
    
    \author{Phillipp Schoppmann}
    \affiliation{%
        \institution{Google}
    	\city{New York City}
    	\state{New York}
    	\country{USA}
    }
    \email{schoppmann@google.com}
    \authornotemark[1]
    
    \renewcommand{\shortauthors}{Ghazi et al.}

\begin{abstract}
The Privacy Sandbox initiative from Google includes APIs for enabling privacy-preserving advertising functionalities as part of the effort around limiting third-party cookies.
In particular, the Private Aggregation API (PAA) and the  Attribution Reporting API (ARA)
can be used for ad measurement while providing different guardrails for safeguarding user privacy, 
including a framework for satisfying differential privacy (DP).
In this work, we provide an abstract model for analyzing the privacy of these APIs and show that they satisfy a formal DP guarantee
under certain assumptions. 
Our analysis handles the case where both the queries and database can change interactively based on previous responses from the API.
\end{abstract}

    
    \maketitle %

\section{Introduction}
Third-party cookies have been a cornerstone of online advertising for more than two decades. 
They are small text files that are stored on a user's computer by websites other than the one
they are currently visiting, allowing websites to track users across the internet and gather
data about their browsing habits, thus enabling advertisers and publisher to measure performance of
ad campaigns.
However, in recent years, growing privacy concerns have led major web browsers to take action: 
both Apple's Safari~\cite{safari3pcd} and Mozilla's Firefox~\cite{mozilla3pcd} deprecated third-party cookies in
$2019$ and $2021$, respectively; Google has made an informed choice proposal to explicitly ask
users to choose whether they want
to disable or enable third-party cookies~\cite{privacySandboxUpdate}. 
This marks a significant change in the online advertising landscape,
increasing the need for new solutions that prioritize user privacy.

In addition to the informed choice proposal, Google has led the Privacy Sandbox~\cite{privacySandbox} initiative,
which includes a set of privacy-preserving technologies aimed at replacing third-party cookie-based ad measurement. 
Two key components of this initiative are the Private Aggregation API (PAA)~\cite{privateAggregationAPI}
and the Attribution Reporting API (ARA)~\cite{attributionReportingAPI},  
which seek to provide advertisers with \emph{reports} that yield insights into the
performance of ad campaigns while protecting user privacy. 
This is already rolled out and activated on roughly $3.5\%$ and $22.1\%$ of all page 
loads in Chrome respectively as of February $1$, $2025$~\cite{PAAStat,ARAStat}.
These APIs offer various privacy guardrails, including a framework~\cite{privacyInAra}
for satisfying differential privacy (DP)~\cite{DworkMNS06,dwork2006our}.

Recall that an advertising (\Ad, for short) campaign is a collection of impressions, each of those 
indicating a user interaction: either a user viewing an ad or clicking on one. 
The websites on which the \Ad{}s are displayed, and possibly clicked, are referred to as \emph{publishers}. 
The goal of an \Ad campaign is to drive useful actions on the advertiser website (e.g., purchases);
these events are usually referred to as \emph{conversions}.
Therefore, goal of measurement for advertiser is to learn about campaign performance to answer questions
like how many users did this \Ad campaign reach or how many converted as a result of this \Ad.

A main component of ARA and PAA are the  \emph{summary reports}~\cite{summaryReportsDev}, which enable
the estimation of aggregate \Ad metrics 
such as the average number of conversions that can be attributed to a certain \Ad campaign~\cite{summaryReportsUseCases},
or the reach of the \Ad campaign~\cite{uniqueReach, whitepaperReach}.  
These reports are usually sliced by some parameters called \emph{keys} 
(e.g., the ad-tech could slice by geography or a device type). 
Summary reports have been proposed in two flavors depending on whether the set of keys over which
the aggregates are sliced is pre-specified or not; the latter is referred to as 
\emph{key discovery}~\cite{keyDiscovery} and arises in practical settings where the set of attributes 
(pertaining to the \Ad, publisher, advertiser and/or user) is very large. 
In addition, it was proposed to extend summary reports to enable \emph{requerying}~\cite{requeryingInAra}, 
which would allow analysts to aggregate their measurement reports on different overlapping slices in an interactive manner 
(i.e., where the output of previous queries can influence the choice of the subsequent queries issued by 
the analyst).

ARA also offers \emph{event-level reports}~\cite{eventReportsDev}, which can provide for each \Ad 
impression, a discretized and noisy estimate of the
list of conversions and conversion values attributed to this impression. 
This type of data enables the training of \Ad conversion models, which are machine learning models
that predict the expected number of conversions (or the conversion value) that would be driven by an \Ad impression.  
The output of these ML models is typically used as input when determining the bid price in the online
auctions that power automated bidding across the Web.\\ 

Despite being deployed in the web browsers of hundreds of millions of users and supporting critical 
\Ad functionalities, the privacy guarantee of neither ARA nor PAA has been formalized rigorously. This gap has been raised in recent 
work~\cite{tholoniat2024alistair,xiao2024click}.

\paragraph{Our Contributions}
In this work, we address this gap and obtain several results formalizing the DP properties of ARA and PAA.
Specifically:
\begin{itemize}
    \item We prove that summary reports in ARA and PAA as well as event-level reports in ARA satisfy a 
        formal DP guarantee even in a highly interactive setting (which captures common practical use
        cases) where the queries and the underlying database can change arbitrarily depending on previous reports.
    \item Our privacy proof for ARA and PAA summary reports applies to the aforementioned proposed 
        extensions that would support key discovery and requerying.
\end{itemize}
Note, however, that both ARA and PAA protect only the \emph{cross-website} (a.k.a. ``third-party'') 
information while the single-website (a.k.a. ``first-party'') information
is assumed to be known to the ad-tech\footnote{%
	In this work, we use the term {\em ad-tech} to refer to an analyst that performs analysis on \Ad{}s and their attributed conversions, thereby helping publishers and advertisers with the placement and measurement of digital ads.
}.  For example, if a user is logged-in on the publisher website,
then the publisher knows the ads shown to the user and 
if a user is logged-in on the advertiser website, then the advertiser knows the user's conversions.  However, if the two websites do not share the same login credentials, the ad-tech will be unable to attribute the conversion to the ad shown, without third party cookies.
The ARA allows the ad-tech to access this attribution information, but in a privacy-preserving manner.

\subsection{Related Work}\label{subsec:related_work}
\paragraph{Differential Privacy.}
Over the last two decades, DP \cite{DworkMNS06, dwork2006our} has become a widely popular notion of privacy in data analytics and modeling, due to its compelling mathematical properties and the strong guarantees that it provides. It has been used in several practical deployments in government agencies (e.g., \cite{census2020, israelNationalRegistry}) and the industry (e.g., \cite{googleMobility}); we refer the reader to \cite{realWorldPrivacy} for a list covering many deployments. As DP is rolled out in ad measurement to replace third-party cookies, it is likely to become one of the largest, if not the largest, real-world deployment of DP, in terms of the number of daily queries and affected users.

\paragraph{Private Ad Measurement.}

In addition to the Privacy Sandbox on Google Chrome and Android, several other APIs have been proposed by various browsers, ad-tech companies, and researchers. 
These include Private Click Measurement (PCM) on Safari~\cite{pcm}, SKAdNetwork on iOS~\cite{skadNetwork}, 
Interoperable Private Attribution (IPA) that was developed by Meta and Mozilla~\cite{ipa}, Masked LARK from Microsoft~\cite{lark}, and 
Cookie Monster which was recently introduced in \cite{tholoniat2024alistair}. All of these APIs, except PCM, use DP to ensure privacy.
We note that our proof, inspired by Cookie Monster \cite{tholoniat2024alistair} is using individual differential privacy accounting. 
Interestingly, while CookieMonster analyzes a similar setting, their notion of adjacency assumes that only impressions, or only conversions are known to the adversary. In other words, CookieMonster setting assumes the ability to hide impressions and/or conversions from the adversary, which may not be compatible with  proposed APIs from most browsers, in particular the ARA and the PAA.

The recent work of \cite{delaney2024differentially} studied the interplay between attribution and DP budgeting for an abstract conversion measurement system. 
It showed that depending on the attribution logic (e.g., first-touch, last-touch, uniform), the privacy unit (e.g., per impression, per user $\times$ publisher, 
per user $\times$ advertiser), and whether contribution capping is performed before or after attribution, 
the sensitivity of the output aggregate can increase with the number of publishers and advertisers rendering the noise too high for accurate measurements.

There has also been recent work on optimizing the utility of the Privacy Sandbox ARA summary reports. 
For hierarchical query workloads, \cite{dawson2023optimizing} gave algorithms for denoising summary reports,
ensuring consistency, and optimizing the contribution budget across different levels of the hierarchy. On the other hand, \cite{aksu2024} presented methods
for optimizing the allocation of the contribution budget for general (not necessarily hierarchical) workloads.

Finally, there has been recent work on DP \Ad (click and conversion) modeling (e.g., \cite{denison2023private, tullii2024position, chua2024private}), 
some of it based on the Privacy Sandbox APIs.

\paragraph{Organization.}
\Cref{sec:ps} provides a simplified description of ARA and PAA; we focus primarily on details relevant for our analysis.
\Cref{sec:formal-notion} provides formal notations, covers basic background on DP and defines the notion of an interactive mechanism and adversary that is relevant for our modeling of ARA and PAA.
In \Cref{sec:summary-reports}, we model summary reports of ARA and PAA and provide the formal DP guarantee for the same.
In \Cref{sec:event-level-reports} we model event-level reports of ARA and provide a formal DP guarantee.%

\section{Privacy Sandbox Measurement API\lowercase{s}}
\label{sec:ps}

We next describe the details of the ARA and PAA APIs and their role in collecting measurements about \Ad{}s and their \emph{attributed conversions}.
(See \Cref{tab:glossary} for a glossary of the involved terminology.)

The \emph{Attribution Reporting API (ARA)} enables ad-techs to measure ad-conversions in a privacy-preserving manner (without third-party cookies). 
In particular, ARA supports two types of reports:
\begin{itemize}[nosep]
	\item \emph{summary reports}, that allow collecting aggregated    
	statistics of \Ad campaigns and their ``attributed'' conversions, and
	\item \emph{event-level reports}, that associate a particular 
	\Ad with very limited (and noisy) data on the conversion side,
	which are sent with a larger time delay.
\end{itemize}

The \emph{Private Aggregation API (PAA)} is another API that also supports {\em summary reports} for collecting aggregated statistics in a privacy-preserving manner. While it is a general API (not necessarily about \Ad{}s),
a typical use-case is in estimating \emph{reach} 
(the number of users who were exposed to an \Ad{}) and \emph{frequency} 
(the number of users that were exposed to an \Ad{} $k$ times, for each $k$).

We first provide a high-level overview of
the \emph{summary reports} in ARA and PAA (in \Cref{sec:summary-ara-paa}), 
and of the \emph{event-level reports} in ARA (in \Cref{sec:summary-event-level}).

\subsection{Summary Reports}
\label{sec:summary-ara-paa}
Summary reports in both ARA and PAA rely on two main components: 
(i) the \emph{client}, which runs in the browser, and
(ii) the \emph{aggregation service}, which runs in a trusted execution environment (TEE)~\cite{teeGit}.
Each client performs local operations based on browser activity and
sends aggregatable reports to ad-techs.  (The exact mechanism that generates these reports is different for ARA and PAA, and we explain this shortly.)  Formally,
an \emph{aggregatable report} is a tuple $(r, k, v, m)$ containing a \emph{key} $k \in \aggRepKeys := \{0, 1\}^{128}$,
a \emph{value} $v \in \{0, 1, \ldots, \contribBudget\}$ (where $\contribBudget := 2^{16}$),
a random identifier $r \in \RID := \{0, 1\}^{128}$ unique to each report 
(the Privacy Sandbox implementation uses AEAD~\cite{AEAD-definition} 
to ensure that the identifier is tamper-proof), and some metadata $m$ that can depend on ``trigger information'' as explained later.  Here, $k$ and $v$ are encrypted, where the secret key for decrypting is held by the aggregation service thereby ensuring that an ad-tech cannot see them and $m$ is visible to the ad-tech in the clear, who can use it to batch the reports for aggregation.  For simplicity, henceforth, we drop $m$ from  aggregate reports  
as it does not affect the privacy analysis.

\subsubsection{Aggregation Service}
Ad-techs can send a subset of the reports, $S = \{(r_1, k_1, v_1), \ldots, (r_n, k_n, v_n)\}$,  they hold (filtered based on the available metadata) to the aggregation service, with the goal of learning, for any $k \in \aggRepKeys$, the aggregated value $w_k = \sum_{i : k_i = k} v_i$.
The aggregation service aggregates the reports to generate a noisy version of this aggregated value.
It supports two modes: one with {\em key discovery} and other without.
\begin{itemize}
\item Without key discovery, the ad-tech needs to provide a subset $L \subseteq \aggRepKeys$ of keys they are interested in, and they
only get corresponding aggregated values $w_\ell$ for $\ell \in L$ after addition of noise sampled from the discrete Laplace distribution, denoted as $\DLap(a)$, which is supported over all integers with probability mass at $x$ proportional to $e^{-a|x|}$.
\item With key discovery, the aggregation service adds noise to each $w_k$ sampled from the {\em truncated} discrete Laplace distribution, $\DLap_{\tau}(a)$, supported over integers in $[-\tau, \tau]$ with probability mass at $x$ proportional to $e^{-a|x|}$. But the noisy values are released only for keys $k$ where these noisy values are greater than $\tau$. The subset $L$ of keys is thus ``discovered'' from the reports themselves.
\end{itemize}
In addition, the aggregation service uses a \emph{privacy budget service} to enforce that the privacy budget
is respected for each aggregatable report.\footnote{%
	In this paper we assume that the aggregation service is tracking budget ``per report'', but in reality it uses a
	coarse id called \emph{shared report id} that consists of the API version, the website that generated the report, 
	and the reporting time in seconds~\cite{sharedReportId} to track a single budget for all reports assigned to the shared report id.
}
This entails maintaining $B_{\eps} : \RID \to \R_{\ge 0}$ that tracks, for each report, the sum of $\eps$ values with which the said report has participated in aggregation requests; the privacy budget service enforces that this sum never exceeds a fixed value $\eps_*$\footnote{%
    Currently this value is equal to $64$~\cite{epsilon64}.
}
Additionally, if using the key discovery mode, an additional state of $B_{\delta} : \RID\to [0, 1]$ is maintained that tracks the number of times a report participated in an aggregation request; and it is enforced that this never exceeds a fixed value $\ell_*$\footnote{%
    Currently this value is equal to $1$ for both ARA and and PAA; this corresponds to no ``re-querying''~\cite{requeryingInAra}.
}
For simplicity we present the aggregation service (with or without key discovery) and the privacy budget service together in \Cref{alg:aggregation-service}. All pseudocode we provide in this paper are for explaining the underlying functionality, and not meant to reflect actual implementation of these APIs.

\newcommand{\nokeydisc}[1]{\textcolor{black!10!blue}{#1}}
\newcommand{\keydisc}[1]{\textcolor{black!30!Ggreen}{#1}}
\begin{algorithm}[t]
	\caption{\AggregationService (with or without Key Discovery)}
	\label{alg:aggregation-service}
	\begin{algorithmic}
		\STATE \algcomment{Parts specific to the service without key discovery are in \nokeydisc{blue}.}
		\STATE \algcomment{Parts specific to the service with key discovery are in \keydisc{green}.}
		\STATE \textbf{Params:} Contribution budget $\contribBudget \in \Z_{> 0}$,
		\STATE \phantom{\textbf{Params:}} privacy parameters $\eps_* \in \R_{\ge 0}$ \keydisc{and $\delta_* \in [0, 1]$}, 
		\STATE \textbf{State:} Privacy budget trackers $B_{\eps} : \RID \to \R_{\ge 0}$ \keydisc{, $B_{\delta} : \RID \to [0, 1]$.}
		\STATE \textbf{Inputs:} Privacy parameters $\eps > 0$ \keydisc{and $\delta \in [0, 1]$},
		\STATE \phantom{\textbf{Input:} } Aggregatable reports 
		$(r_1, k_1, v_1)$, \dots, $(r_n, k_n, v_n)$,
		\STATE \phantom{\textbf{Input:} } 
		\nokeydisc{Subset $L = \{\ell_1, \dots, \ell_m\} \subseteq \aggRepKeys$ of keys.}
		\STATE \textbf{Output:} Summary report $(\ell_1, w_1)$, \dots, $(\ell_m, w_m)$.
		\IF{$\exists i \in [n]$ such that $B_{\eps}(r_i) + \eps > \eps_*$ \keydisc{or $B_{\delta}(r_i) + \delta > \delta_*$}} 
		\STATE {\bf Abort}\algcomment{Privacy budget violated for some report.}
		\ELSE
		\STATE $B_{\eps}(r_i) \gets B_{\eps}(r_i) + \eps$ for all $i \in [n]$
		\STATE \keydisc{$B_{\delta}(r_i) \gets B_{\delta}(r_i) + \delta$ for all $i \in [n]$}
		\ENDIF
		\STATE \keydisc{Let $L \subseteq K$ be the set of distinct keys among $k_1$, \dots, $k_n$}
		\STATE \keydisc{$\tau \gets \contribBudget \cdot \left (1 + \log(\sparsityBudget / \delta) / \eps \right)$}
		\FOR{$\ell \in L$}
		\STATE $w_\ell \gets z_\ell + \sum_{j ~:~ k_j = \ell} v_j$ \quad for \ $z_\ell \sim \DLap_{\keydisc{\tau}}(\eps / \contribBudget)$
		\ENDFOR
		\RETURN $\{(\ell, w_\ell) : \ell \in L \keydisc{\text{ and } w_\ell > \tau}\}$
	\end{algorithmic}
\end{algorithm}

\begin{remark}%
    We note that the currently supported implementation of the aggregation service is weaker than our description
    in that the key discovery mode is not supported and $\ell_* = 1$ is enforced for all reports 
    (that is, each report can participate in at most one aggregation request).
    However, our results hold even under these proposed extensions of key discovery~\cite{keyDiscovery} and requerying~\cite{requeryingInAra}.
\end{remark}

Next, we describe the clients for ARA and PAA that generate aggregatable reports based on cross-site information and send them to the ad-tech.

\subsubsection{ARA Client}\label{subsubsec:ara-client}
We explain the workings of $\ARASRClient$ (the client that supports ARA summary reports in the browser) in context of an example 
visualized in \Cref{fig:aggregatable-report-example}. 
Suppose a publisher's page (\texttt{\small some-blog.com}) displays \Ad{}s as part of a ``Thanksgiving'' 
campaign about sneakers, sandals, and flip-flops that can be purchased on an advertiser's page 
(\texttt{\small shoes-website.com}). 
The ad-tech would like to measure the amount of money spent on shoes on the advertiser's page
that could be attributed to these ads, and in particular, with an attribution rule that ``the purchase must happen within three days after an ad was shown''.

To use $\ARASRClient$, the ad-tech annotates both the publisher and advertiser pages with additional 
JavaScript or HTML that points to a URL with a specific HTTP header 
(see \cite{attributionReportingAPI} for details). 
$\ARASRClient$ can get invoked in two ways, either when the \Ad is shown on the publisher's page, 
or when a conversion happens on the advertiser's page as described below.

\begin{figure*}
	\centering
	\includegraphics[width=0.8\linewidth]{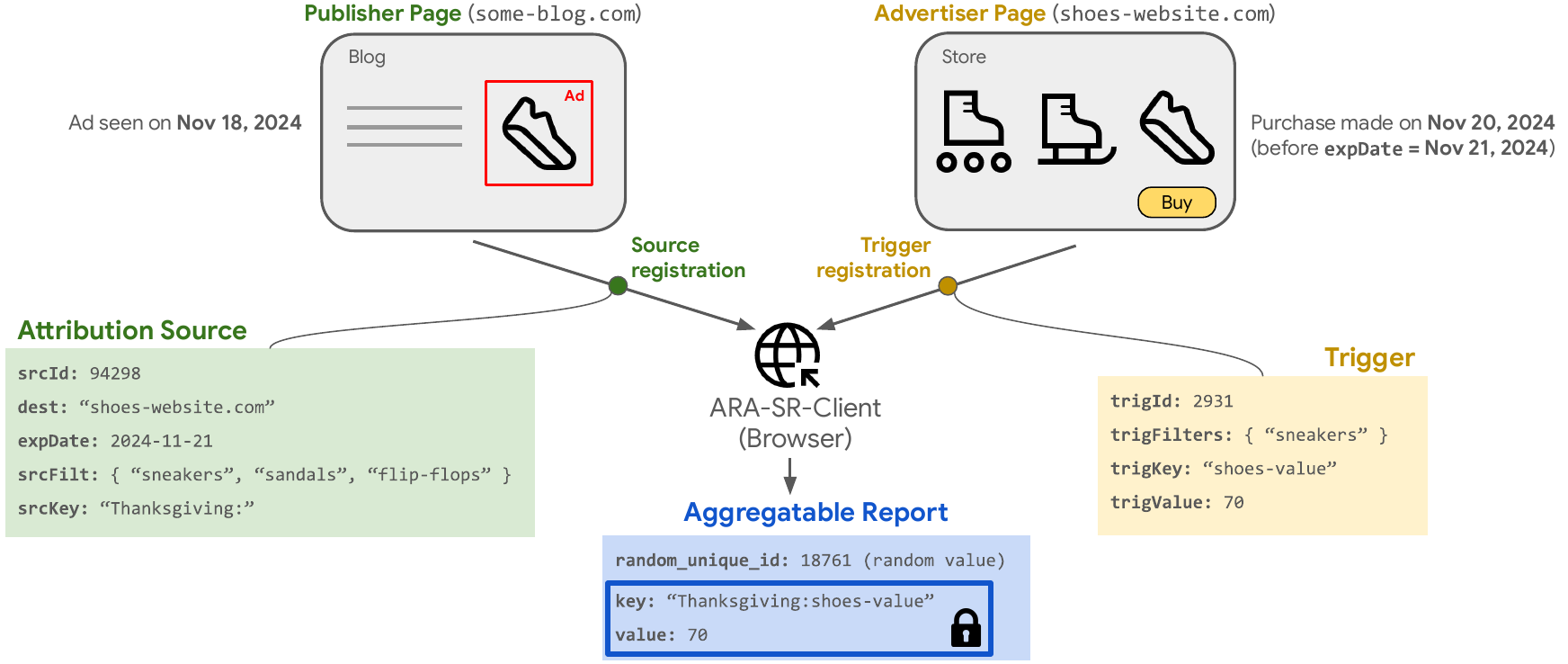}
	\caption{
		Illustrative example of an aggregatable report generated by \ARASRClient.
	}
	\label{fig:aggregatable-report-example}
\end{figure*}

\begin{itemize}
	\item
	On the \emph{publisher page}, the ad-tech registers a so-called ``{\em attribution source}'' with the $\ARASRClient$ in the browser (e.g., corresponding to a view/click ad event). In our example the source corresponds to an ad shown on \texttt{\small some-blog.com}. This entails specifying a tuple
	$(\sourceId, \destination, \expDate, \sourceFilters, \sourceKey)$, where
	\begin{itemize}
		\item $\sourceId$ is an identifier associated to the source event 
		(in our example it is a random
        id generated when an ad is shown on \texttt{\small some-blog.com}),
		\item $\destination$ is the advertiser domain where the conversion could happen 
		(\texttt{\small shoes-website.com} in our example),
		\item $\expDate$ is an expiration date for attributing conversions to the source 
		(three days in our example),
		\item $\sourceFilters$ is a set of filter keys, each being a bounded bit string,
		(treated as strings ``sneakers'', ``sandals'', and ``flip-flops'' in our example for simplicity),
		\item \sourceKey $\in \aggRepKeys$ is a key associated with the source (in our example we take it to be the string ``{\tt Thanksgiving:}'' for illustration).
	\end{itemize}
	Each source registration invokes the ``Source registration'' part of \Cref{alg:ara-client}, 
        which updates a set $\activeSources$
        of registered sources.
	\item 
	On the \emph{advertiser page}, the ad-tech registers a so-called ``\emph{trigger}'' corresponding to a qualifying user activity (such as a purchase on \texttt{\small shoes-website.com} in our example)
	by specifying a tuple
	$(\triggerId, \triggerFilters, \triggerKey, \triggerValue)$,
	where
	\begin{itemize}
		\item \triggerId is an identifier associated with the trigger 
		(in our example it is a random id generated when the purchase happened on \texttt{\small shoes-website.com}),
		\item \triggerFilters is a set of filter keys similar to \sourceFilters 
		(in our example it could be either ``sneakers'', or ``sandals'', or ``flip-flops''),
		\item \triggerKey $\in \aggRepKeys$ is a key associated to the trigger (in our example we take it to be the string ``{\tt shoes-value}'' for illustration).
		\item \triggerValue  $\in \{0, 1, \ldots, \contribBudget\}$ is a value associated to the trigger
		(in our example it is the price of the shoes purchased).
	\end{itemize}
	Each trigger registration invokes the ``Trigger registration'' part of \Cref{alg:ara-client}.
\end{itemize}

At any point of time, $\ARASRClient$ uses
$L_1 : \activeSources \to \Z_{\ge 0}$ to track the sum of all values contributed, and
$L_0 : \activeSources \to \Z_{\ge 0}$ to denote the number of non-zero values contributed per registered source, with $L_0(s)$ and $L_1(s)$ initialized to $0$ at the time of registration of source $s \in \activeSources$.

When a user visits the advertiser's page and a trigger is registered, $\ARASRClient$
(\Cref{alg:ara-client}) matches the trigger to the most recently registered {\em active} source (namely, those with $\expDate$ ahead of the current time) and 
generates an aggregatable report; aggregatable reports for ARA contain \triggerId
as part of its metadata (earlier referred to as $m$ in $(r, k, v, m)$). If no source is matched, a null report, denoted $(r, \bot, \bot)$ in \Cref{alg:ara-client} is sent. Such reports are ignored by the aggregation service~(\Cref{alg:aggregation-service}).

\begin{remark}
In reality the metadata $m$ and the time the report is sent can allow the ad-tech to associate the received report with the particular device. Moreover, if the trigger does not get attributed to any source, and $\triggerId$ is not set, the corresponding null report is not sent (see \cite{nullReports}). This could potentially allow the ad-tech to know if a trigger was attributed to a source or not. These are handled by the ARA client in practice with some heuristics such as sending reports without a trigger id with some delay, as well as sending some fake null reports with small probability. However, these heuristics would not allow us to prove a formal DP guarantee, so we do not consider this case.
\end{remark}

\begin{remark}
	$\ARASRClient$ (\Cref{alg:ara-client}) is using the so-called ``last-touch'' attribution, namely, 
    that in presence of several potentially matching sources, the most recently registered one is chosen.
	While we choose to only consider last-touch attribution for simplicity, 
    our results hold for any attribution method (as long as any trigger is fully attributed to only one source). 
    Indeed, ARA uses a more involved attribution strategy where impressions
    can be assigned priority (at source registration) and the trigger is attributed to the last source with the highest priority.
\end{remark}

\begin{algorithm}[t]
\caption{\ARASRClient}
\label{alg:ara-client}
\begin{algorithmic}
\STATE \textbf{Params:} Contribution and sparsity budgets $\contribBudget, \sparsityBudget \in \Z_{> 0}$.
\STATE \textbf{State:}
\begin{itemize}[topsep=0pt]
\item Set of registered sources $\activeSources$,
\item Sparsity Budget Tracker $L_0 : \activeSources \to \Z_{\ge 0}$,
\item Contribution Budget Tracker $L_1 : \activeSources \to \Z_{\ge 0}$.\vspace{1mm}
\end{itemize}
\STATE \hrule
	\STATE \textbf{Source registration:} 
	\STATE On input $s = (\sourceId, \destination, \expDate, \sourceFilters, \sourceKey)$
	\STATE $\activeSources \gets \activeSources \cup \{s\}$ \algcomment{Add $s$ to set of registered sources.}
	\STATE $L_0(s) = 0$ and $L_1(s) = 0$\vspace{1mm}
\STATE \hrule
\STATE {\bf Trigger registration:}
\STATE On input $(\destination, \triggerId, \triggerFilters, \triggerKey, \triggerValue)$.
\STATE $r \gets$ a random report id in $\RID$
\IF{$\triggerValue > 0$}
	\FOR{active $s \in \activeSources$ (in reverse chronological order)}
		\IF{$\destination = s.\destination$ and  $\triggerFilters \cap s.\sourceFilters \neq \emptyset$}
			\IF{$L_0(s) + 1 \le \sparsityBudget$ and 
				$L_1(s) + \triggerValue \le \contribBudget$}
				\STATE $v \gets \triggerValue$
				\STATE $k \gets $ bit-wise OR of \sourceKey and \triggerKey
				\STATE $L_0(s) \gets L_0(s) + 1$
				\STATE $L_1(s) \gets L_1(s) + \triggerValue$
				\STATE {\bf\boldmath Send report $(r, k, v)$ and halt}
			\ENDIF
		\ENDIF
	\ENDFOR
\ENDIF
\STATE {\bf\boldmath Send null report $(r, \bot, \bot)$}
\end{algorithmic}
\end{algorithm}

Here is how our example would play out, as visualized in \Cref{fig:aggregatable-report-example}.
A publisher's page (\texttt{\small some-blog.com}) displays an \Ad, on Nov 18th about shoes that can be purchased on an advertiser's
page (\texttt{\small shoes-website.com}). When an ad is displayed at \texttt{\small some-blog.com} 
the source gets registered with a random $\sourceId$, with 
$\destination$ equal to $\texttt{\small ``shoes-website.com''}$, an expiry date of Nov 21st, that is three days into the future from the time the ad was shown,
a set of $\sourceFilters= \texttt{\small \{``sneakers'', ``sandals'', ``flip-flops''\}}$
that restricts which purchases can be attributed to
this source and, a source key $\sourceKey=\texttt{\small ``Thanksgiving:''}$ to tie this ad to a certain campaign.

When the user subsequently purchases a sneaker from the advertiser's page, the advertiser can choose to register a trigger with a random $\triggerId$, 
a set of $\triggerFilters=\texttt{\small\{``sneakers''\}}$ that restricts the potential sources that this conversion can be attributed to, a trigger key 
$\triggerKey=\texttt{\small``shoes-value''}$ to record the interpretation of the value, and finally $\triggerValue=70$, which is the price of the sneaker.

The \ARASRClient (\Cref{alg:ara-client}) then attributes the trigger to the source and generates an aggregatable report $(r, k, v)$ with the key being 
$k=\texttt{``Thanksgiving:shoes-value''}$ (where for ease of illustration we use string concatenation instead of bitwise-OR), value $v=70$, and $r$ being a 
random id.

Recall that the ad-tech knows details of the \Ad impression that was displayed on the publisher website,
as well as details of the conversion that happened on the advertiser's website; these are referred to as ``first-party'' information. 
However, without third-party cookies, the ad-tech cannot link the two events as happening 
on the same browser and thus cannot attribute the conversion to the impression. 
The ARA helps ad-techs access this ``third-party'' information linking the conversion to the impression. 
But it does so in a privacy-preserving manner by enforcing that the total number of generated reports 
associated to any source is at most $\sparsityBudget$, and that the total value of such generated reports is at most $\contribBudget$.
We provide the formal privacy guarantees implied by these properties of ARA summary reports as \Cref{cor:ara-sr-dp} in \Cref{sec:summary-reports}.

\subsubsection{PAA Client \& Shared Storage}\label{subsubsec:paa-client}
\begin{figure*}
\centering
\includegraphics[width=0.85\linewidth]{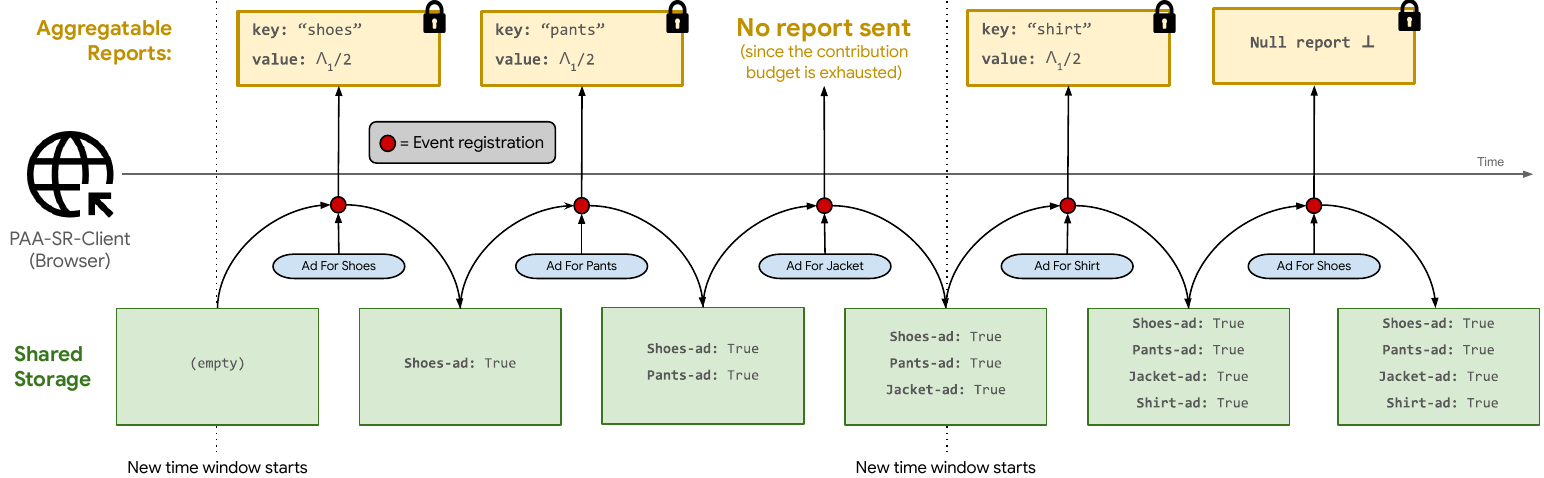}
\caption{Illustrative example of an aggregatable report generated by \PAAClient.}
\label{fig:paa-aggregatable-report-example}
\end{figure*}
The PAA client is typically used in tandem with the Shared Storage API~\cite{shared-storage}.
Shared Storage is a key-value database stored on the browser that can be accessed by the code using PAA across any
websites visited on a browser, which the ad-tech can annotate with their code.

Recall that $\ARASRClient$ tracks the contribution and sparsity budgets for each ``source''. 
However, unlike ARA, the PAA does not have a concept of registering a source. 
PAA instead enforces contribution and sparsity budgets separately for each ad-tech
and each ``time-window''.
Formally, let $\deviceSet$ denote the set of all devices, and let $\timeWindowSet$ denote the partition of all time into contiguous time windows, 
each of a fixed length.\footnote{Each time window is 10 minutes long in the current implementation of PAA~\cite{10minutes}.}
We use $\Xi$ to denote the shared storage, where $\Xi_\device$ denotes the part of storage accessible to the device $\device \in \deviceSet$.
At any point of time, $\PAAClient$ uses $L_1 : \deviceSet \times \timeWindowSet \to \Z_{\ge 0}$ to track the sum of all values contributed and $L_0 : \deviceSet \times \timeWindowSet \to \Z_{\ge 0}$
to track the number of non-zero values contributed per $(\device, t)$, with $L_0(\device,t)$ and $L_1(\device, t)$ initialized to $0$ at the start of time window $t$ for any $\device \in \deviceSet$.

\begin{algorithm}[t]
    \caption{\PAAClient}
    \label{alg:paa-client}
    \begin{algorithmic}
        \STATE \textbf{Params:} Contribution and sparsity budgets $\contribBudget, \sparsityBudget \in \Z_{>0}$.
        \STATE \textbf{State:}
        \begin{itemize}[topsep=0pt]
            \item Shared storage $\Xi \in \mathfrak{X}^{\deviceSet}$ (we use $\mathfrak{X}$ to denote the set of possible states of the shared storage, one for each device $\device \in \deviceSet$),
            \item Sparsity Budget Tracker $L_0 : \deviceSet \times \timeWindowSet \to \Z_{\ge 0}$,
            \item Contribution Budget Tracker $L_1 : \deviceSet \times \timeWindowSet \to \Z_{\ge 0}$.\vspace{1mm}
        \end{itemize}
        \STATE \hrule
        \STATE {\bf Event Registration}
        \STATE On following inputs:
        \begin{itemize}[topsep=0pt]
        \item The device $\device$ registering the event,
        \item Time window $t$ when event is invoked, and
        \item Function $\pi : \mathfrak{X} \to \mathfrak{X} \times \aggRepKeys \times \{0, 1, \ldots, \contribBudget\}$ 
            that on input being the current state of shared storage $\Xi$, returns its next state in addition to a key value pair $(k, v)$.
        \end{itemize}
        \STATE $\Xi_\device, k, v \gets \pi(\Xi_\device)$ \algcomment{Note: Ad-tech does not learn what is inside $\Xi_a$. Moreover, even the returned $(k, v)$ is not visible to ad-tech.}
        \STATE $r \gets$ a random report id in $\RID$
        \IF{$v > 0$}
        	\IF{$L_0(a, t) + 1 \le \sparsityBudget$ and $L_1(a, t) + v \le \contribBudget$}
        		\STATE $L_0(\device,t) \gets L_0(\device, t) + 1$
        		\STATE $L_1(\device, t) \gets L_1(\device, t) + v$
        		\STATE {\bf\boldmath Send aggregatable report $(r, k, v)$ and halt}
        	\ENDIF
        \ENDIF
        \STATE {\bf\boldmath Send null report $(r, \bot, \bot)$}
    \end{algorithmic}
\end{algorithm}

We explain the working on $\PAAClient$ in context of an example 
visualized in \Cref{fig:paa-aggregatable-report-example}. 
Assume an ad-tech has four \Ad campaigns (for shoes, pants, jackets, and shirts) 
running in parallel and they want to estimate the reach of each campaign 
(the number of people exposed to the ads from campaigns) across different publisher websites. 
Suppose the ad-tech expects that most of the people would only see at most two ads within a single time window $t \in \timeWindowSet$.

To use the PAA and Shared Storage APIs, the ad-tech annotates any webpage it has access to with a JavaScript that can read and write to Shared Storage and can register an ad {\em event};
the registration requires providing a method $\pi$ that maps the current state of the shared storage, to the next state of the storage, a key $k \in \aggRepKeys$, and 
a corresponding  value $v \in \{0, 1, \ldots, \contribBudget\}$ 
(in our example the key is the campaign name and the value is $\contribBudget / 2$ if the ad from that campaign was seen for the first time and $0$ otherwise; this is looked up from shared storage). 
When an event is registered, $\PAAClient$ (\Cref{alg:paa-client}) checks that the new addition would not violate the contribution and sparsity budget for the ad-tech at the current time window,
and if that is indeed the case, it generates an aggregatable report $(r, k, v)$ for a random ID $r \in \RID$.

To run through our example, the user sees first the ads for shoes and pants and the reports
are sent; next the user sees an ad for a jacket, but the contribution budget is exhausted so no 
report is sent. Finally, in the next time window, the user sees an ad for a shirt and the report
is sent since the budget is refreshed. However, when the user sees an ad for shoes again, a null report is sent because the contribution is $0$, since shared storage indicates that a shoes ad was seen previously.

Similar to case of ARA, recall that the ad-tech knows details of any specific event in PAA, which is considered ``first-party'' information. However, without third-party cookies, the ad-tech cannot link information across events happening on the same browser across different websites; this would be ``third-party'' information.
The PAA helps ad-techs access this ``third-party'' information via the use of the Shared Storage API.
But it does so in a privacy-preserving manner by enforcing that the total number of generated reports for any ad-tech in any time window is at most $\sparsityBudget$, and that the total value of such generated reports is at most $\contribBudget$. The access to Shared Storage is also ``sandboxed'' in a manner that the information within can only be used for purposes of generating the aggregatable report.
We provide the formal privacy guarantees implied by these properties of PAA summary reports as \Cref{cor:paa-sr-dp} in \Cref{sec:summary-reports}.

\subsection{Event-level Reports}
\label{sec:summary-event-level}

In addition to summary reports, ARA also supports {\em event-level} reports~\cite{eventLevelReports}.
Before describing the event-level reports formally, we consider an example visualized in \Cref{fig:event-level-report-example}. 
Again, suppose a publisher's webpage (\texttt{\small some-blog.com}) displays an ad for footwear
sold on an advertiser's site (\texttt{\small shoes-website.com}).

To use $\ARAEventClient$, the ad-tech annotates both the publisher and advertiser pages
in a manner that is similar to case of summary reports (see \Cref{subsubsec:ara-client}). $\ARAEventClient$ can get invoked in three ways, (i) when the \Ad is shown on publisher's page, (ii) when a conversion happens on the advertiser's page, and (iii) on the passing of a ``reporting window'', as explained below (this happens automatically without any action from the ad-tech). For simplicity, we first describe the ``noiseless'' version of $\ARAEventClient$ in \Cref{alg:noiseless-event-level-report}.

\begin{figure*}
\centering
\includegraphics[width=0.99\linewidth]{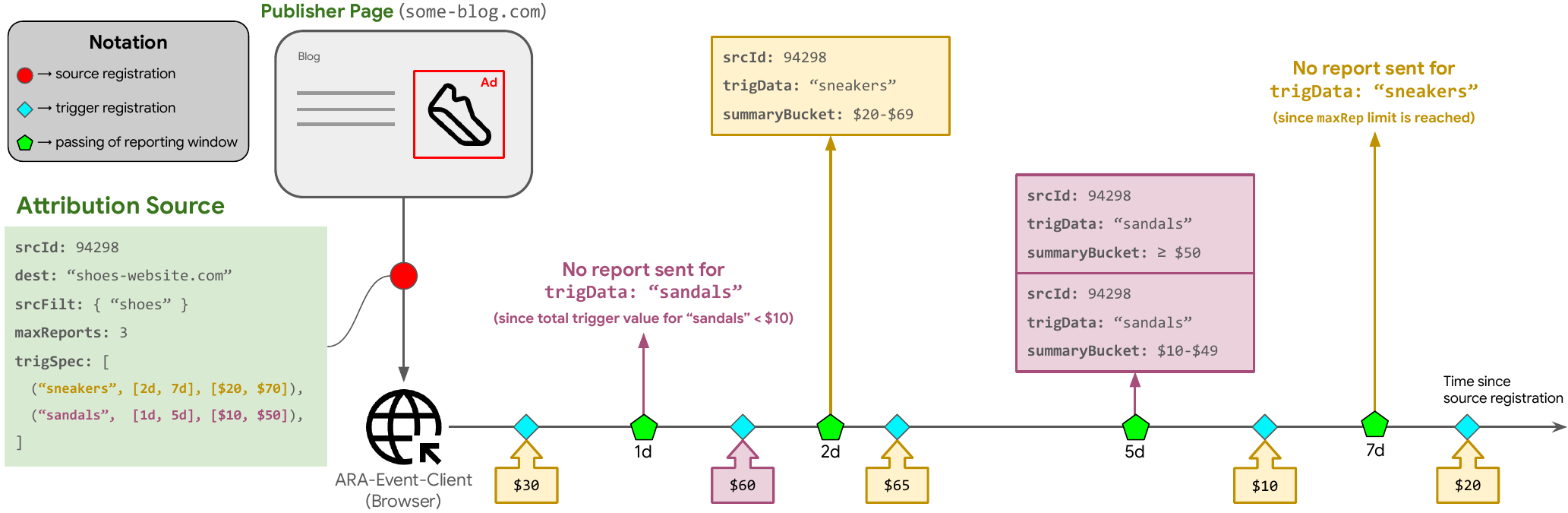}
\caption{Illustrative example of event-level reports generated by \ARAEventClient.}
\label{fig:event-level-report-example}
\end{figure*}

\begin{itemize}
\item 
On the \emph{publisher page}, the ad-tech registers an {\em attribution source} with $\ARAEventClient$ in the browser. %
This entails specifying a tuple 
$
(\sourceId, \destination, \expDate, \sourceFilters, \maximalNumberOfReports, \triggerSpecification)
$,
where
\begin{itemize}
	\item \sourceId, \destination, \expDate, \sourceFilters are the same as in the case of the ARA summary reports,
	\item \maximalNumberOfReports is the maximum number of event-level reports that can be attributed to this source 
	(in our example it is 3),
	\item $\triggerSpecification$,\footnote{%
		In real implementation a list of possible values of trigger data is passed separately and there
		is a default specification, but we ignore this here for simplicity.
	} short for ``trigger specificiation'', is a list of at most 32 elements (in our case we have two specifications, one for ``sneakers'' and 
	one for ``sandals'') each consisting of 
	\begin{itemize}
		\item $\triggerData \in \triggerDataSet$, where $\triggerDataSet$ is the set of possible values that $\triggerData$ can take; while this is restricted to be a $5$-bit integer, we treat $\triggerData$ as a string in our example for ease of visualization.
		\item at most five \emph{reporting windows} that specify when the new reports should be sent; 
		each reporting window is time (in seconds) 
		from the source registration time.
		\item \emph{summary buckets}, an 
		increasing list of integer thresholds 
		such that a report is sent when the threshold is met (subject to \maximalNumberOfReports limit).
	\end{itemize}
\end{itemize}
Each source registration invokes the ``Source registration'' part of \Cref{alg:noiseless-event-level-report}; visualized by the red dot in \Cref{fig:event-level-report-example}. This adds the source to the list of active registered sources $\activeSources$ and initializes the current values and last reported bucket values to zero.

\item 
On the \emph{advertiser page}, the ad-tech registers a {\em trigger} corresponding to a qualifying user activity (such as a purchase on \texttt{\small shoes-website.com} in our example) 
by specifying a tuple containing $(\destination, \triggerId, \triggerFilters, \triggerData, \triggerValue)$, 
where
\begin{itemize}
	\item \destination, \triggerId, \triggerFilters are the same as in the case of ARA summary reports,
	\item \triggerData is data associated with the trigger, which must match one from \triggerSpecification 
	(if it does not match any in the specification it gets ignored), and
	\item \triggerValue is an integer value associated with the trigger (in our example the amount spent on the shoes).
\end{itemize}
Each trigger registration invokes the ``Trigger registration'' part of \Cref{alg:noiseless-event-level-report}; visualized by the cyan diamond in \Cref{fig:event-level-report-example}. Here $\ARAEventClient$ searches for an active source that this trigger could be attributed to, and upon finding one, adds the value of the current trigger to the value associated to the source for this $\triggerData$.
\end{itemize}

\begin{algorithm}[t]
\caption{\ARAEventClient (Noiseless version).}
\label{alg:noiseless-event-level-report}
\begin{algorithmic}
\STATE \textbf{State:}
\STATE $\bullet$ Sequence of active registered sources $\activeSources$, %
\STATE $\bullet$ Current value $V : \activeSources \times \triggerDataSet \to \Z_{\ge 0}$,
\STATE $\bullet$ Last reported bucket $U : \activeSources \times \triggerDataSet \to \Z_{\ge 0}$,
\STATE $\bullet$ Number of reports sent $N : \activeSources \to \Z_{\ge 0}$.\vspace{1mm}
\STATE \hrule
	\STATE \textbf{Source registration:} 
	\STATE On input $s = (\sourceId, \destination, \expDate, \sourceFilters, \maximalNumberOfReports, \triggerSpecification)$
	\STATE $\activeSources \gets \activeSources \cup \{s\}$ \algcomment{Add $s$ to set of active sources.}
	\STATE $N(s) \gets 0$
	\FOR{$\triggerData \in \triggerDataSet$}
		\STATE $V(s, \triggerData) \gets 0$ and $U(s, \triggerData) \gets 0$ \vspace{1mm}
	\ENDFOR
\STATE \hrule
	\STATE \textbf{Trigger registration:}
	\STATE On input $(\destination, \triggerId, \triggerFilters, \triggerData, \triggerValue)$
	\IF{$\triggerValue > 0$}
		\FOR{$s \in \activeSources$ (in reverse chronological order)}
			\IF{$\destination = s.\destination$ and  $\triggerFilters \cap s.\sourceFilters \neq \emptyset$}
				\IF{$\triggerData$ listed in $s.\triggerSpecification$}
					\STATE $V(s, \triggerData) \gets V(s, \triggerData) + \triggerValue$
					\STATE \algcomment{Add to value associated to this source and $\triggerData$.}
					\STATE {\bf break}\vspace{1mm}
				\ENDIF
			\ENDIF
		\ENDFOR
	\ENDIF
\STATE \hrule
	\STATE \textbf{\boldmath At passing of time window for $(s, \triggerData) \in \activeSources \times \triggerDataSet$:}
	\STATE Let $b_1, \ldots, b_k \in \Z_{> 0}$ be the list of associated summary buckets to $\triggerData$ in $\triggerSpecification$ of source $s$, in increasing order.
	\FOR{$i=1,\ldots, k$}
		\IF{$U(s, \triggerData) < b_i \le V(s, \triggerData)$}
			\STATE $U(s, \triggerData) = b_i$
			\IF{$N(s) < s.\maximalNumberOfReports$}
				\STATE $N(s) \gets N(s) + 1$
				\STATE {\bf\boldmath Send report $(s.\sourceId, \triggerData, b_i)$ to ad-tech}
			\ENDIF
		\ENDIF
	\ENDFOR
\end{algorithmic}
\end{algorithm}

\noindent Finally, for each registered source $s$, whenever the amount of time passed equals the reporting window for any $\triggerData$ in $s.\triggerSpecification$, a report is sent to the ad-tech using \Cref{alg:noiseless-event-level-report}; in our example, $1$ day after the source is registered, the first reporting window for $\triggerData = \texttt{\small ``sandals''}$ is passed and $\ARAEventClient$ get invoked. All such invocations associated to passing of reporting windows is visualized as a green pentagon in \Cref{fig:event-level-report-example}.

We now walk through the behavior of $\ARAEventClient$ in the example visualized in \Cref{fig:event-level-report-example}; we present the ``noiseless'' behavior of the algorithm before describing the noisy version.
\begin{itemize}
	\item \textbf{Day 0:} An ad is shown on publisher's page and the source $s$ gets registered.
	\item \textbf{Day 1:} The user buys \$30 sneakers from advertiser's page. This is attributed to source $s$ and $V(s, \texttt{\small sneakers})$ is set to $30$.
	
	At the end of day 1, the reporting window for $\texttt{\small sandals}$ is passed, but no report is generated because $V(s, \texttt{\small sandals})$ is still $0$.
	\item \textbf{Day 2:} The user buys \$60 sandals. This is attributed to source $s$, and $V(s, \texttt{\small sandals})$ is set to $60$.
	
	At the end of day 2, the reporting window for $\texttt{\small sneakers}$ is passed, and a report is generated indicating a bucket value of \$20-\$69. At this point $U(s, \texttt{\small sneakers})$ is $20$.
	\item \textbf{Day 4:} The user buys \$65 sneakers. This is attributed to source $s$ and $V(s, \texttt{\small sneakers})$ gets incremented to become $95$.
	\item \textbf{Day 5:} At the end of day 5, the second reporting window for $\texttt{\small sandals}$ is passed, and two reports are generated for sandals as $V(s, \texttt{\small sandals}) = 60$, which exceeds both summary bucket values of \$10 and \$50.
	\item \textbf{Day 6:} The user buys \$10 sneakers. This is attributed to source $s$ and $V(s, \texttt{\small sneakers})$ is incremented to become $105$.
	\item \textbf{Day 7:} At the end of day 7, the second reporting window for $\texttt{\small sneakers}$ is passed, but no reports are sent, as the maximum number of reports $3$ for this source has been reached. Furthermore, the source now gets marked as inactive and no further purchases get attributed to this source.
\end{itemize}

\paragraph{Noisy Version of $\ARAEventClient$.}
Note that for each source, the set $\rroutputvalid$ of possible combinations of reports that could be sent is finite.
In our example, reports are sent for \texttt{\small ``sandals''} only at the end of Day 1 and Day 5. The number of reports sent on each of the two days is one among the $6$ possible options shown in \Cref{tab:event-options}.
\begin{wrapfigure}{r}{0.16\textwidth}
\centering\vspace{-3mm}
\begin{tabular}{|c|c|}
\hline
{\bf Day 1} & {\bf Day 5} \\
\hline
0 & 0\\
0 & 1 \\
0 & 2 \\
1 & 0 \\
1 & 1 \\
2 & 0\\
\hline
\end{tabular}
\captionof{table}{Possible number of reports sent on Day 1 and Day 5 for \texttt{\small ``sandals''}.}
\label{tab:event-options}
\vspace{-8mm}
\end{wrapfigure}
Similarly, the number of reports sent at the end of Day 2 and Day 7 is also among $6$ possible options. While these lead to $36$ possible configurations of reports sent, $9$ of them correspond to sending of $4$ reports, which is larger than the maximum limit of $3$ reports that could be sent. Thus, in total, there are $27$ possible valid configurations of reports that could get sent.
The private version of event-level reports works as follows: 
with probability $\frac{e^{\eps} - 1}{e^{\eps} + |\rroutputvalid| - 1}$ 
we proceed like in the noiseless version in \Cref{alg:noiseless-event-level-report}, otherwise we sample an element of $\rroutputvalid$ uniformly at
random, and generate responses accordingly. This choice can be made at the time of source registration, even before any triggers are observed.
We provide the formal privacy guarantees implied by these properties of ARA event-level reports as \Cref{theorem:event-level-privacy} in \Cref{sec:event-level-reports}.

\section{Differential Privacy}
\label{sec:formal-notion}

We use the notion of differential privacy (DP), which typically considers \emph{mechanisms} $\mechanism : \cD \to \Delta(\ResponseSet)$ that map input datasets $D \in \cD$ to probability distributions over a set  $\ResponseSet$ of responses. Central to the notion of DP, is the definition of ``adjacency'' of databases. Loosely speaking two databases $D, D' \in \cD$ are said to be adjacent
if they ``differ in one record''. We defer the definition of a database and the notion of adjacencies relevant for the analysis in each application to \Cref{sec:summary-reports}.
But for any notion of database and adjacency, differential privacy (DP) quantifies the ability (or lack thereof) of an adversary to distinguish two ``adjacent'' datasets $D$ and $D'$ by observing a sample from $\mechanism(D)$ or $\mechanism(D')$.

\begin{definition}[$(\eps, \delta)$-Indistinguishability]\label{def:DP-indistinguishable}
Two distributions $P$, $Q$ are said to be \emph{$(\eps, \delta)$-indistinguishable}, denoted $P \approx_{\eps, \delta} Q$ if
for all events $W$, it holds that
\[
P(W) \le e^\eps Q(W) + \delta \ \text{ and } \ Q(W) \le e^\eps P(W) + \delta.
\]
\end{definition}

\begin{definition}[Differential Privacy]\label{def:dp}\cite{DworkMNS06, dwork2006our}
A mechanism $\mechanism : \DataSets \to \Delta(\ResponseSet)$ satisfies \emph{$(\eps, \delta)$-DP} 
if for all adjacent datasets $D, D' \in \DataSets$, it holds that $\mechanism(D) \approx_{\eps, \delta} \mechanism(D')$.
The special case of $(\eps, 0)$-DP is denoted as $\eps$-DP for short.
\end{definition}

We will often refer to parameter $\eps$ or both $\eps$ and $\delta$ in the definition of $(\eps, \delta)$-DP as 
the \emph{privacy budget}.

However, as we discuss shortly, modeling the setting of ARA and PAA as mechanisms operating on a ``static'' database does not capture the interactive nature of these systems. Hence we consider a stronger notion of ``interactive mechanisms'' and ``interactive adversaries'' defined in \Cref{subsec:interactive-mechs}.

Since the aggregation service adds noise from the (truncated) discrete Laplace distribution, we rely on the following fact to prove DP properties of summary reports (for both ARA and PAA). For any integer $d \in \Z_{> 0}$, let $\DLap_{\tau}(a)^{\otimes d}$ be the distribution over $\Z^d$ where each coordinate is drawn independently from $\DLap_{\tau}(a)$.
\begin{fact}[\cite{GRS12}]\label{fact:dlaplace-dp}
    For all $\eps > 0$ and integer $\Delta > 0$, and vectors $u, v \in \Z^d$ such that $\|u - v\|_1 = \Delta$, 
    the distributions $P$ and $Q$ drawn as $u + \xi$ and $v + \xi$ for $\xi \sim \DLap(\eps/\Delta)^{\otimes d}$
    satisfy $P \approx_{\eps, 0} Q$.

    Furthermore, if $u - v$ has at most $s$ non-zero coordinates, then for all $\delta > 0$ and 
    $\tau \ge \Delta \cdot (1 + \log(s/\delta) / \eps)$,
    the distributions $P$ and $Q$ drawn as $u + \xi$ and $v + \xi$ for $\xi \sim \DLap_{\tau}(\eps/ \Delta)^{\otimes d}$ 
    satisfy $P \approx_{\eps, \delta} Q$.
\end{fact}

While the privacy guarantee of the discrete Laplace mechanism was studied in \cite{GRS12}, we include a proof for the case of truncated discrete Laplace noise in \Cref{apx:tdlaplace} for completeness.

\subsection{Interactive Mechanisms}
\label{subsec:interactive-mechs}

As noted earlier, DP is typically defined for ``one-shot'' mechanisms that map databases to probability distributions over a response set. However, there are two key reasons why the generation process of ARA/PAA summary reports is not a ``one-shot'' mechanism:
\begin{enumerate}
	\item \label{item:adaptivity}
	There is adaptivity in the choice of contribution values, which can be changed based on previously observed summary reports. For example, in context of ARA, the ad-tech can change the scale and interpretation of $\triggerValue$ when making contributions based on previously observed summary reports.
	\item \label{item:online}
	The ad-tech can have some influence on the new events that get added to the database themselves.
	For example, in the context of ARA, the advertiser might decide based on summary reports obtained
	on day $1$ to change the product price on day $2$, which could influence the number of subsequent 
	conversions.
\end{enumerate}
To argue the DP properties of summary and event-level reports, it is helpful to model their generation process as an interaction between a mechanism and an adversary defined below. Abstractly speaking, let $\DataSets$ denote the set of all ``databases'', and let $\QuerySet$ and $\ResponseSet$ denote a set of ``queries'' and ``responses'' as used below.
\begin{definition}[Interactive Mechanism] 
    \label{def:interactive-system}
    An \emph{interactive mechanism} with state set $\StateSet$ is represented by an
    initial state $\mechState_0 \in \StateSet$ and a function 
    $\mechanism : \StateSet \times \DataSets \times \QuerySet \to \StateSet \times \Delta(\ResponseSet)$, 
    that maps the current state $S \in \StateSet$, a database $D \in \DataSets$ and a query $q \in \QuerySet$ to the next state $S' \in \StateSet$ and a distribution over responses $r \in \ResponseSet$.%
\end{definition}

We will often abuse notation to denote the output of $\mechanism(S, D, q)$ as $(S', r)$ where $S'$ is the next state and $r$ is drawn from the distribution over $\ResponseSet$ returned by $\mechanism(S, D, q)$.

\begin{definition}[Interactive Adversary]
\label{def:interactive-adversary}
    An \emph{interactive adversary} $\adversary : \ResponseSet^* \to (\DataSets \times \QuerySet) \cup \{\halt\}$ maps the history of responses $(r_1, r_2, \ldots) \in \R^*$, to the next database $D \in \DataSets$ and query $q \in \QuerySet$, or ``halt'' ($\halt$).
\end{definition}

The interaction between $\mechanism$ with initial state $\mechState_0 \in \StateSet$, and an interactive
adversary $\adversary$ is described in \Cref{alg:general-transcript} and results in the probability distribution
$\transcript{\mechanism}{\adversary}$ of \emph{transcripts} $\Pi$, which is a sequence of responses $(r_1, r_2, \ldots) \in \ResponseSet^*$.

\begin{algorithm}[t]
\caption{Interactive Transcript $\transcript{\mechanism}{\adversary}$.}
\label{alg:general-transcript}
\begin{algorithmic}
\STATE \textbf{Inputs:} $\triangleright$ Interactive mechanism $\mechanism$ with initial state $\mechState_0$,
\STATE \phantom{\textbf{Inputs:}} $\triangleright$ Interactive adversary $\adversary$ 
\STATE $\Pi \gets ()$ and $t \gets 1$ \algcomment{Empty transcript.}
\WHILE{$\adversary(\Pi) \ne \halt$}
\STATE $(D_t, q_t) \gets \adversary(\Pi)$ \algcomment{Adversary creates database \& query.}
\STATE $(\mechState_t, r_t) \sim \mechanism(\mechState_{t-1}, D_{t}, q_t)$ \algcomment{Mechanism samples response.}
\STATE $\Pi \gets \Pi \circ r_t$ and $t \gets t + 1$
\ENDWHILE
\RETURN $\Pi$
\end{algorithmic}
\end{algorithm}

In order to define what it means for an interactive mechanism to satisfy DP, we need to define the notion of ``adjacency'' for databases.
For now, let us abstractly say that database is a set of records $(x, y) \in \cX\times \cY$.
The set $\cX$ is assumed to be known to the adversary and we will refer to $x \in \cX$ as a ``privacy unit''.
Let $\cY$ be an arbitrary set for the sake of our notation here; we instantiate it appropriately as relevant later.
For any database $D$, let $D^{-x}$ denote the database that loosely speaking ``removes records in $D$ corresponding to $x \in \cX$ or replaces them with a certain generic one''.
We leave this notion to be abstract for now, and we instantiate the specific notion of adjacency when describing the formal guarantees of ARA and PAA. 
For any interactive mechanism $\mechanism$ and any $x \in \cX$, let $\mechanism^{-x}$ denote the mechanism that replaces the dataset $D_t$ by $D_t^{-x}$ at each step: that is, $(S_t, r_t) \sim \mechanism(S_{t-1}, D_t, q_t)$ is 
replaced by $(S_t, r_t) \sim \mechanism(S_{t-1}, D_t^{-x}, q_t)$ in \Cref{alg:general-transcript}. We define DP for interactive mechanisms as follows.

\begin{definition}[DP for Interactive Mechanisms]\label{def:xDP}
    An interactive mechanism $\mechanism$ satisfies $(\eps, \delta)$-DP
    if for all interactive adversaries $\adversary$ and all $x \in \cX$, it holds that
    $\transcript{\mechanism}{\adversary} \approx_{\eps, \delta} \transcript{\mechanism^{-x}}{\adversary}$.
\end{definition}

\begin{remark}
A single round version of our definition coincides with the standard definition of DP for the adjacency notion where $D$ and $D^{-x}$ are adjacent.
For multiple rounds, our definition is strictly stronger than the standard definition of DP for interactive mechanisms,
where the adversary always returns the same database at each step; we refer to such adversaries as ``stable''.
\end{remark}

\section{Analysis of Summary Reports}
\label{sec:summary-reports}

\subsection{DP guarantees from IDP}

For any interactive mechanism $\mechanism$ and a sequence of databases and queries $((D_1, q_1), (D_{2}, q_{2}), \ldots)$, let $\function_t(\cdot)$ denote the distribution over response $r_t$ as returned by $\mechanism(S_{t-1}, \cdot, q_t)$; note that the sequence $(S_0, S_1, \ldots, S_{t-1})$ of mechanism states is deterministic given the sequence of databases and queries. 
We say that the sequence $((\eps_1, \delta_1), (\eps_2, \delta_2), \ldots)$ is a {\em privacy rollout} of the mechanism $\mechanism$ for $x \in \cX$ on the sequence $((D_1, q_1), (D_2, q_2), \ldots)$, if $\function_t(D_t) \approx_{\eps_t, \delta_t} \function_t(D_t^{-x})$ holds for all $t$.

\begin{definition}[Individual DP]\label{def:eps-delta-IDP}
An interactive mechanism $\mechanism$ satisfies {\em $(\eps_*, \delta_*)$-IDP} if for all $x \in \cX$ and all sequences $((D_1, q_1)$, $(D_2, q_2)$, $\ldots)$ of databases and queries, if $((\eps_1, \delta_1), (\eps_2, \delta_2), \ldots)$ is a privacy rollout of $\mechanism$ for $x$ on the said input sequence, then $\sum_t \eps_t \le \eps_*$ and $\sum_t \delta_t \le \delta_*$.%
\end{definition}

\noindent Our main technical result is that IDP implies DP. This can be viewed as the approximate-DP variant of  \cite[Theorem 4.5]{FZ21}, which proves a qualitatively similar statement for R\'enyi DP although our result is for the more general case of interactive adversaries, wherein even the database can change based on previous responses. We note however that our proof technique is quite general, and can be applied in the R\'enyi DP setting to extend the result of \cite{FZ21} to the case of interactive adversaries as well.

\begin{theorem}\label{thm:individual-DP-to-DP}
    If an interactive mechanism satisfies $(\eps_*, \delta_*)$-IDP, then it satisfies $(\eps_*, \delta_*)$-DP.
\end{theorem}

To prove the above, we apply the tool of DP filters~\cite{rogers16privacy} in the setting of IDP~\cite{ESS15,FZ21}.

We defer the full proof to \Cref{subsec:idp-to-dp-proof},
and first describe how this implies the DP guarantees of the ARA and PAA summary reports by appropriately instantiating the notion of a database and adjacency as well as the notion of queries and responses and showing that the interactive mechanism that generates the corresponding summary reports satisfies an IDP guarantee and hence by the above result, also satisfies a DP guarantee.

\subsection{Privacy of ARA Summary Reports}\label{subsec:ara-sr-privacy}

To prove the DP properties of ARA, we formalize an end-to-end mechanism $\MSR$ (\Cref{alg:summary-reports-interactive-mechanism}) that simulates the joint behavior of $\ARASRClient$ (\Cref{alg:ara-client}) and the aggregation service (\Cref{alg:aggregation-service}) ultimately generating the summary reports.

\paragraph{Databases and Adjacency.} We model a database $D\in \cD_\text{ARA}$ as consisting of records 
$(x, y) \in \cX_\bot \times \cY$
where $\cX_\bot = \cX \cup \{\dummyImpression\}$ is the set $\cX$ of all possible 
``sources'' registered across all devices
in addition to a ``dummy source'' that we denote as $\dummyImpression$ and $\cY$ is the set of all possible
``triggers'' registered across all devices; note that triggers are in one-to-one correspondence with the report ID $r$ part of the generated aggregatable report $(r, k, v)$, and thus for simplicitly, we interchangeably use $y \in \cY$ to denote the report ID.
For any database $D \in \DataSets$ and $x \in \cX$, let $D^{-x}$ be the dataset obtained by moving all aggregatable
reports associated to $x$ to instead be associated with $\dummyImpression$,  that is, replace $(x, y)$ by $(\dummyImpression, y)$.
We note that adjacent databases in our notion have the same set of $x$'s and $y$'s, and thus, DP is not protecting against knowledge of these, but only the knowledge of which $y$'s are attributed to which $x$'s.

\paragraph{Queries, Responses and Mechanism States.} The query set $\cQ$ for $\MSR$ consists of tuples $q = (\eps, \delta, Y, f)$ where $(\eps, \delta) \in \R_{\ge 0} \times [0, 1]$ are privacy parameters for the query, $Y \subseteq \cY$ is a subset of triggers whose corresponding reports need to be aggregated and $f: \cX \times \cY \to \aggRepKeys \times \N$ is a function mapping a pair
of a source and a trigger to the candidate (key, value) for the aggregatable report generated by them.
The response set $\ResponseSet$ of $\MSR$ is the set of summary reports, namely $(\aggRepKeys \times \Z)^*$.
Finally, the state set $\StateSet$ of $\MSR$ is given by the tuple $(\{ (L_x, s_x) \}_{x \in \cX}, 
\{ (\eps_y, \delta_y) \}_{y \in \cY}, R)$.  Here, 
$L_x$ (resp., $s_x$) is the sum of all (resp., number of non-zero) contributions attributed to $x \in \cX$, 
$\epsilon_y, \delta_y$ are privacy budgets consumed for each report (equivalently trigger) $y \in \cY$, and $R$ is the set of all aggregatable reports generated so far.

\paragraph{Relating $\MSR$ to ARA Client and Aggregation Service.} For $\MSR$ to simulate the end-to-end generation of summary reports by ARA, we instantiate the database $D$ at each step to be the set of {\em new} impressions and trigger pairs registered since the last query to $\MSR$. In phase \#1, $\MSR$ updates the set $R$ in its state to have all the aggregatable reports generated so far, applying the contribution and sparsity bounding similar to $\ARASRClient$ (\Cref{alg:ara-client}). In phase \#2, $\MSR$ tracks and enforces that the privacy budget used per report specified in $Y$ is under limits. If not, it aborts. Else it updates these privacy budgets per report and in phase \#3, returns the noisy summation per key (with or without key discovery as specified). Phases \#2 and \#3 simulate the aggregation service (\Cref{alg:aggregation-service}).

\begin{algorithm}[t]
\caption{Interactive mechanism 
$\MSR : \StateSet \times \cD \times \QuerySet \to \StateSet \times \Delta(\ResponseSet)$.}
\label{alg:summary-reports-interactive-mechanism}
\begin{algorithmic}
\STATE \textbf{Params:} $\triangleright$ 
        Contribution budget $\contribBudget$,
        Sparsity budget $\sparsityBudget$,
\STATE \phantom{\textbf{Params:}} $\triangleright$ Global privacy parameters $(\eps_*, \delta_*)$.
\STATE \textbf{State:} $\triangleright$
        $\{ (L_x, s_x) \}_{x \in \cX}$,
	$\{ (\eps_y, \delta_y) \}_{y \in \cY}$, and $R \subseteq \cY \times \aggRepKeys \times \N$.
\STATE \textbf{Inputs:} $\triangleright$ Database $D \in \cD_\text{ARA}$,
\STATE \phantom{\textbf{Inputs:}} $\triangleright$ 
      Privacy parameters $\eps > 0$ and $\delta \in [0, 1]$,
\STATE \phantom{\textbf{Inputs:}} $\triangleright$ The list of triggers $Y \subseteq \cY$ whose corresponding
\STATE \phantom{\textbf{Inputs:} $\triangleright$} reports are to be aggregated, and
\STATE \phantom{\textbf{Inputs:}} $\triangleright$ 
    The function $f: \cX \times \cY \to \aggRepKeys \times \N$

\STATE {\color{black!70} \# 1. Contribution Bounding Per $x \in \cX$ (cf. $\ARASRClient$)}
\FOR{$(x, y) \in D$}
    \STATE $(k, v) \gets f(x, y)$
    \IF{$L_x + v \le \contribBudget$ and $s_x + 1 \le \sparsityBudget$}
        \STATE $L_x \gets L_x + v$
        \STATE $s_x \gets s_x + 1$
        \STATE $R \gets R \cup \{(y, k, v)\}$
    \ENDIF
\ENDFOR

\STATE
\STATE {\color{black!70} \# 2. Budget Bounding Per $y \in \cY$ (cf. Aggregation Service)}
\IF{$\exists y \in \cY$ such that $\eps_y + \eps > \eps_* \text{ or } \delta_y + \delta > \delta_*$}
    \RETURN (($\{ (L_x, s_x) \}_{x \in \cX}, \{ (\eps_y, \delta_y) \}_{y \in \cY}, R$), {\bf abort})
\ELSE
	\STATE $\eps_y \gets \eps_y + \eps$ and $\delta_y \gets \delta_y + \delta$ for all $y \in Y$
\ENDIF

\STATE 
\STATE {\color{black!70} \# 3. Noisy summation (cf. Aggregation Service)}
\STATE $\tau \gets \begin{cases}
	\infty & \text{if } \delta = 0\,,\\
	\contribBudget \cdot (1 + \log(\sparsityBudget / \delta) / \eps) & \text{if } \delta > 0\,.
\end{cases}$

\STATE $S \gets \emptyset$
	\FOR{$k \in \aggRepKeys$}
	      \STATE $c_k \gets \xi + \sum_{(y, k', v) \in R \, : \, y \in Y,\ k' = k} v$\ \ for\ \ $\xi \sim \DLap_{\tau}(\eps / \contribBudget)$
		\IF{$\tau = \infty$ or $c_k > \tau$}
		      \STATE $S \gets S \cup \{ (k, c_k) \}$
		\ENDIF
	\ENDFOR
\STATE
\RETURN $\left( ( \{ (L_x, s_x) \}_{x \in \cX}, \{ (\eps_y, \delta_y) \}_{y \in \cY}, R), S \right)$
\end{algorithmic}
\end{algorithm}

\begin{theorem}\label{thm:ara-sr-idp}
    $\MSR$ satisfies $(\eps_*, \delta_*)$-IDP with respect to aforementioned set of databases and neighbouring relation for ARA.
\end{theorem}

Before proving \Cref{thm:ara-sr-idp}, we note that putting this together with \Cref{thm:individual-DP-to-DP}, immediately implies the following corollary.
\begin{corollary}\label{cor:ara-sr-dp}
    $\MSR$ satisfies $(\eps_*, \delta_*)$-DP with respect to aforementioned set of databases and neighbouring relation for ARA.
\end{corollary}

\begin{proof}[Proof of \Cref{thm:ara-sr-idp}]
    Let $(D_1, (\eps_1, \delta_1, Y_1, f_1)), \dots, $ be any sequence of databases and queries that is provided to $\MSR$;
    assume the state at step $t$ is 
    $(\{ (L_{t, x}, s_{t, x}) \}_{x \in \cX}$, $\{ (\eps_{t, y}, \delta_{t, y}) \}_{y \in \cY}, R_t)$. We assume without loss of generality that there is no round where $\MSR$ returns a response of {\bf abort}. This is because whether $\MSR$ outputs {\bf abort} or not is only a function of the $y$ values in the database, and is the same for $D$ and $D^{-x}$.
    
    Consider any $x \in \cX$. Denote by $R_{x, t} \subseteq R_t$ the set of all aggregatable 
    reports generated for $x$. Let $Y_{x, t}$ denote the set of all $y \in Y_t$ such that $(y, k_y, v_y) \in R_{x, t}$ for some $(k_y, v_y)$ and let $Y_x := \bigcup_{t} Y_{x,t}$.
    Note that $\sum_{y \in Y_{x}} v_y \le \contribBudget$ and $|Y_{x}| \le \sparsityBudget$ for all $t$.

    Thus, from the privacy guarantee of the discrete Laplace mechanism~(\Cref{fact:dlaplace-dp}) we have that
    $\MSR(S_{t - 1}, D_t, (\eps_t, \delta_t, Y_t, q_t)) \approx_{\eps_{x,t}, \delta_{x,t}} \MSR(S_{t - 1}, D_t^{-x}, (\eps_t, \delta_t, Y_t, q_t))$
    where
    \begin{align*}
        \eps_{x,t} := \frac{\eps_t}{\contribBudget} \cdot \sum_{y \in Y_{x,t}} v_y
        \qquad \text{and} \qquad
        \delta_{x,t} := \frac{\delta_t}{\sparsityBudget} \cdot |Y_{x,t}|,
    \end{align*}
    and each step $t$. Thus, we get that $\MSR$ satisfies $(\eps_*, \delta_*)$-IDP, since for any $x \in \cX$, it holds that
    \begin{align*}
        \sum_t \eps_{x,t} &= \sum_t \left(\frac{\eps_t}{\contribBudget} \cdot \sum_{y \in Y_{x,t}} v_y\right)
        ~\le~ \sum_{y \in Y_x} \frac{v_y}{\contribBudget} \cdot \sum_{t \ :\ Y_{x,t} \ni y} \eps_t
        ~\le~ \eps_*,
    \end{align*}
    where the last inequality follows because $\sum_{t : Y_{x,t} \ni y} \eps_t \le \eps_*$ for all $y \in \cY$ and $\sum_{y \in Y_x} v_y \le \contribBudget$ for all $x \in \cX$.  Similarly,
    \begin{align*}
        \sum_t \delta_{x,t} &= \sum_t \left(\frac{\delta_t}{\sparsityBudget} \cdot |Y_{x,t}| \right)
        ~\le~ \sum_{y \in Y_x} \frac{1}{\sparsityBudget} \sum_{t \ :\ Y_{x,t} \ni y} \delta_t ~\le~ \delta_*,
    \end{align*}
    where the last inequality follows because $\sum_{t : Y_{x,t} \ni y} \delta_t \le \delta_*$ for all $y \in \cY$ and $|Y_x| \le \sparsityBudget$ for all $x \in \cX$.
\end{proof}

\subsection{Privacy of PAA Summary Reports}\label{subsec:paa-sr-privacy}

The DP properties of PAA can be proved with the same formalism of the interactive mechanism $\MSR$~(\Cref{alg:summary-reports-interactive-mechanism}), with just a reinterpretation of databases and the adjacency notion.

In the context of PAA summary reports, we model a database $D\in \cD_\text{PAA}$ as consisting of records 
$(x, y) \in \cX \times \cY_\bot$ where $\cX$ is the set of all pairs consisting of (device, $t$) for 
$t \in \timeWindowSet$ is the time window and $\cY_\bot$ is the set of all possible 
states of the shared storage $\cY$ and a dummy shared storage $\dummySharedStorage$ (denoting an ``empty'' shared storage).
For any database $D \in \DataSets$ and $x \in \cX$, let $D^{-x}$ be the dataset obtained by replacing all shared storage states associated to $x$ by $\dummySharedStorage$.
The notion of queries, responses and states remain the same as in the case of ARA.

The following theorem has essentially the same proof as \Cref{thm:ara-sr-idp}, so we do not repeat it.
\begin{theorem}\label{thm:paa-sr-idp}
    $\MSR$ satisfies $(\eps_*, \delta_*)$-IDP with respect to aforementioned set of databases and neighbouring relation for PAA.
\end{theorem}

Hence, a corollary similar to \Cref{cor:ara-sr-dp} holds.
\begin{corollary}\label{cor:paa-sr-dp}
    $\MSR$ satisfies $(\eps_*, \delta_*)$-DP with respect to aforementioned set of databases and neighbouring relation for ARA.
\end{corollary}

Note however, that it might be desirable to have privacy guarantees per device, but the above only provides a privacy guarantee per (device, time window $t$). One could try to obtain a device-level DP guarantee by using the ``group privacy'' property of DP.

\begin{fact}[Group Privacy~\cite{vadhan17complexity}]\label{fact:group-privacy}
If distributions $P_0, P_1, \ldots, P_k$ are such that $P_i \approx_{\eps, \delta} P_{i+1}$ for all $i$, then the distributions $P_0 \approx_{\eps', \delta'} P_k$ for $ \eps' = k\eps$ and $\delta' = \delta \frac{e^{k\eps}-1}{e^\eps - 1}$.
\end{fact}

However, it is tricky to apply this property because if the shared storage is disabled or cleared in a certain time window $t$, it can affect the future states of shared storage and thereby also affect the aggregatable reports generated by events on that device.

One can however formulate a weaker notion of device-level DP guarantee for PAA, by considering a variant of DP with gradual expiration that is considered in \cite{AHPSU24}.
\begin{definition}[DP for Interactive Mechanisms with Gradual Expiration]\label{def:xDP-with-expiration}
    An interactive mechanism $\mechanism$ satisfies \emph{$(\eps, \delta)$-DP with gradual expiration} if
    for all interactive adversaries $\adversary$, all $e \in \Z$, all devices $\device \in \deviceSet$ and all 
    $t_1, t_2 \in \cT$ with $t_1 < t_2$, the distributions 
    $\transcript{\mechanism^{-(\device, >t_1)}}{\adversary} \approx_{\eps', \delta'} \transcript{\mechanism^{-(\device, >t_2)}}{\adversary}$ 
    where $\eps' = \eps (t_2 - t_1)$, $\delta' = \delta \frac{e^{\eps'} - 1}{e^{\eps} - 1}$, and
    $\mechanism^{-(a, >t_1)}$ denotes the mechanism that replaces the dataset $D_t$ by $D_t^{-(\device, >t_1)}$ 
    (obtained by replacing all shared storage states associated to $(\device, t)$ for $t > t_1$ by $\dummySharedStorage$) at each step.
\end{definition}

It is immediate to see that if $\mechanism$ satisfies $(\eps, \delta)$-DP, then it satisfies $(\eps, \delta)$-DP with gradual expiration via \Cref{fact:group-privacy}. 
Thus, we get that in the context of PAA, the mechanism $\MSR$ satisfies $(\eps_*, \delta_*)$-DP with gradual expiration.

\subsection{Proof of \texorpdfstring{\Cref{thm:individual-DP-to-DP}}{Theorem~\ref{thm:individual-DP-to-DP}}}\label{subsec:idp-to-dp-proof}

As mentioned before, to prove \Cref{thm:individual-DP-to-DP}, we apply the tool of DP filters~\cite{rogers16privacy} in the setting of IDP~\cite{ESS15,FZ21}, that we discuss in \Cref{subsec:filter,subsec:indv-priv} below respectively.

\subsubsection{Privacy Filters}
\label{subsec:filter}

To discuss privacy filters, consider databases $U \subseteq \cX$ that contain only privacy units.
We use $U^{-x}$ to denote the dataset $U \smallsetminus \set{x}$.\footnote{%
	The notation $U^{-x}$ is useful only when $x \in U$. 
	Otherwise $U^{-x} = U$, which can lead to vacuous statements.
	Since these are also correct we do not enforce that $x \in U$.
}

For $\Theta = \R_{\ge 0} \times [0, 1]$ being the parameter space underlying $(\eps, \delta)$-DP, let $\filter : \Theta^* \to \{\thumbsup, \thumbsdown\}$
be a function, called {\em filter}, that maps a sequence of parameters to $\thumbsup$ or $\thumbsdown$.
For any $\ResponseSet$, and $\cP(\cX)$ denoting the powerset of $\cX$, we consider the following notion of a 
\emph{universal interactive mechanism} 
$\mechanism_{\filter} : \StateSet \times \cP(\cX) \times \QuerySetU \to \StateSet \times \Delta((\ResponseSet \cup \{\bot\}))$, 
parameterized by $\filter$, that performs on-the-fly privacy budgeting defined as follows:
\begin{itemize}
	\item $\QuerySetU := \{ \query: \cP(\cX) \to \Delta(\ResponseSet) \}$ consists of ``universal queries'' 
	that can be arbitrary ``one-shot'' mechanisms; an example of such a mechanism is 
	$q(U) = \sum_{x \in U} v_x + \DLap(a)$, where $v_x$ and $a$ are chosen as part of the query,
	\item $\StateSet = \Theta^*$ consists of sequences of privacy parameters.
\end{itemize}
On query $q$ such that $q$ satisfies $\theta$-DP,\footnote{$q$ may satisfy $\theta$-DP 
	for a number of different $\theta \in \Theta$, and while any such $\theta$ could be used, it is 
	important to specify a particular choice of $\theta$ to make $\mechanism_{\filter}$ well-defined. 
	The definition of $\cQ$ could be modified to additionally provide the specific $\theta$ along with $q$, 
	but we avoid doing so for simplicity.} $\mechanism_{\filter}$ operates as follows:
	\begin{multline*}
		\mechanism_{\filter}((\theta_1, \ldots, \theta_t), D, q)\\
		:=
		\begin{cases}
			((\theta_1, \ldots, \theta_t, \theta), q(D)) & \text{if } \filter(\theta_1, \ldots, \theta_t, \theta) = \thumbsup\\
			((\theta_1, \ldots, \theta_t, \mathbf{0}), \bot) & \text{if } \filter(\theta_1, \ldots, \theta_t, \theta) = \thumbsdown
		\end{cases}
	\end{multline*}

That is, if $\filter$ applied on the current state (sequence of $\theta_i$'s so far) concatenated with
the current $\theta$ returns $\thumbsup$, then $\theta$ is concatenated to the current state, and the 
one-shot mechanism $q$ is applied on $U$. 
But if not, then $\mathbf{0}$ is concatenated to the state and $\bot$ is returned.

The interactive mechanism $\mechanism_{\filter}$ interacts with an adversary 
$\adversary$
in the same way as in \Cref{alg:general-transcript} to produce $\transcript{\mechanism_\filter}{\adversary}$.

For $(\eps, \delta) \in \Theta$, we define the filter $\filter_{\eps,\delta}: \Theta^* \to \{\thumbsup, \thumbsdown\}$ as:
\[
    \filter_{\eps, \delta}((\eps_1, \delta_1), \dots, (\eps_n, \delta_n)) := \begin{cases}
    	\thumbsup & \text{if } \sum\limits_{i = 1}^n \eps_i \le \eps \text{ \& } \sum\limits_{i = 1}^n \delta_i \le \delta\\
    	\thumbsdown & \text{otherwise.}
    \end{cases}
\]

It is known that $\mechanism_{\filter}$ with the above filter satisfies DP guarantees.
\begin{lemma}[\cite{rogers16privacy}]\label{lem:filters}
	For all $(\eps, \delta) \in \Theta$, $\mechanism_{\filter_{\eps, \delta}}$ satisfies $(\eps, \delta)$-DP against stable adversaries.%
	\footnote{%
		\cite{rogers16privacy} also provides an improved ``advanced composition''-like filter for 
		$(\eps, \delta)$-DP, that was subsequently improved in~\cite{whitehouse23fully}. 
		However, since we only use this version, we do not state the advanced version.
	}
\end{lemma}

\subsubsection{Individual Differential Privacy}
\label{subsec:indv-priv}

We now consider a universal interactive mechanism with a privacy filter but using the notion of IDP~\cite{ESS15,FZ21}.

\begin{definition}
	\label{def:individual-DP}
	For $p : \cX \to \Theta$, a mechanism $\function: \cP(\cX) \to \Delta(\ResponseSet)$ satisfies 
	\emph{$p$-IDP} if for all $U \in \cP(\cX)$ and all $x \in \cX$, it holds that $\function(U) \approx_{p(x)} \function(U^{-x})$.%
\end{definition}

As before, for $\filter : \Theta^* \to \{\thumbsup, \thumbsdown\}$, we consider a universal interactive mechanism 
$\indmechanism_{\filter} : \StateSet \times \cP(\cX) \times \QuerySetU \to \StateSet \times \Delta(\ResponseSet)$,
that performs on-the-fly privacy budgeting, where
$\StateSet = (\Theta^{\cX})^*$ consists of sequences of privacy parameters, one for each unit $x \in \cX$.
On query $q$ such that $q$ satisfies $p$-IDP, $\indmechanism_{\filter}$ operates as follows. 
On current state $(p_1, \ldots, p_t)$, it constructs a ``masked database''
$\overline{U} := U \cap \{ x : \filter(p_1(x), \ldots, p_t(x), p(x)) = \thumbsup \}$
and the consumed individual privacy $p_{t+1} : \cX \to \Theta$ as:
\[
p_{t+1}(x) := \begin{cases}
	p(x) & \text{if } \filter(p_1(x), \ldots, p_t(x), p(x)) = \thumbsup\\
	\mathbf{0} & \text{if } \filter(p_1(x), \ldots, p_t(x), p(x)) = \thumbsdown.
\end{cases}
\]
And the mechanism returns,
\[
\mechanism_{\filter}((p_1, \ldots, p_t), U, q) := ((p_1, \ldots, p_{t+1}), q(\overline{U})).
\]

\begin{lemma}\label{lem:indv-filter}
	For all $(\eps, \delta) \in \Theta$, $\indmechanism_{\filter_{\eps, \delta}}$ satisfies $(\eps, \delta)$-DP against stable adversaries.
\end{lemma}
\begin{proof}
	Fix some database $U \subseteq \cX$.
	Consider a stable adversary $\adversary : \ResponseSet^* \to (\cP(\cX) \times \QuerySetU) \cup 
	\{\halt\}$ that returns the same database $U$ on each step.
	We want to show that for any $x \in \cX$ it holds that
	$\transcript{\indmechanism_{\filter_{\theta}}}{\adversary} \approx_{\eps, \delta} \transcript{\left(\indmechanism_{\filter_{\theta}}\right)^{-x}}{\adversary}$.
	This follows by ``zooming in on unit $x$''. Namely, consider $\cX' = \set{x}$, and let 
	$\QuerySetU' := \set{q : \cX' \to \Delta(\ResponseSet)}$ be the set of universal queries on $\cX'$. 
	Construct an adversary $\adversary' : \ResponseSet^* \to (\cP(\cX') \times \QuerySetU') \cup 
	\{\halt\}$ as follows: 
	$\adversary'(\Pi)$ first computes $(D, q) \gets \adversary(\Pi)$, and return $(D', q')$ where $D' = D \cap \cX'$ and $q'(U') := q(\tilde{U}^{-x} \cup U')$ for 
	$U' \subseteq \cX'$ and $\tilde{U}^{-x} \subseteq U^{-x}$ that have the privacy budget to participate.
	Thus, when $U \ni x$, we have
	$\transcript{\indmechanism_{\filter_{\theta}}}{\adversary} \equiv 
	\transcript{\mechanism_{\filter_{\theta}}}{\adversary'} \approx_{\eps, \delta} 
	\transcript{\left(\mechanism_{\filter_{\theta}}\right)^{-x}}{\adversary'} \equiv
	\transcript{\left(\indmechanism_{\filter_{\theta}}\right)^{-x}}{\adversary}$, by \Cref{lem:filters}.
\end{proof}

\subsubsection{Putting it Together: Proof of \texorpdfstring{\Cref{thm:individual-DP-to-DP}}{Theorem~\ref{thm:individual-DP-to-DP}}}

\begin{proof}[Proof of \Cref{thm:individual-DP-to-DP}]
    Let $\mechanism$ be an interactive mechanism satisfying $(\eps_*, \delta_*)$-IDP guarantees.
    To prove the statement it suffices to show that for any adversary $\adversary$ that interacts with $\mechanism$, there is an adversary $\adversary'$ for
    $\mechanism_{\filter_{\eps_*, \delta_*}}$ such that
    \begin{align*}
        \transcript{\indmechanism_{\filter_{\eps_*, \delta_*}}}{\adversary'} &\equiv \transcript{M}{\adversary}, \text{ and}\\
        \transcript{\left(\indmechanism_{\filter_{\eps_*, \delta_*}}\right)^{-x}}{\adversary'} &\equiv \transcript{M^{-x}}{\adversary}.
    \end{align*}
    Conditioned on a sequence  $(r_1, \ldots, r_t)$ of responses, let the corresponding states of $\mechanism$ when interacting with $\adversary$ be $S_0, S_1, \ldots, S_{t}$; note $S_i$ is deterministic given database $D_i$ and query $q_i$, which are in turn deterministic given $r_1, \ldots, r_{i-1}$.
    
    We define $\adversary'$ that on input $(r_1, \ldots, r_t)$, computes $(D_{t+1}, q_{t+1}) \gets \adversary(r_1, \ldots, r_t)$ and returns $(D_{t+1}, q'_{t+1})$ where $q'_{t+1}(\tilde D)$ is distribution over responses $r$ returned by $\mechanism(S_t, \tilde D, q_{t+1})$.
    It is easy to see that $\adversary'$ satisfy the condition since $\mechanism_{\filter_{\eps_*, \delta_*}}$ will always output $q'(D_{t+1})$ given  $\mechanism$ is $(\eps_*, \delta_*)$-IDP, that is, the filter never masks any element of the database; 
    hence from \Cref{lem:indv-filter}, we conclude that $\transcript{M}{\adversary} \approx_{\eps_*, \delta_*} \transcript{M^{-x}}{\adversary}$.
\end{proof}

\section{Analysis of Event-Level Reports}
\label{sec:event-level-reports}
The database and the privacy unit in event-level reports is the same as described in \Cref{subsec:ara-sr-privacy}.

The event-level API is based on what we refer to as
\emph{interactive randomized response (IRR)}; see \Cref{alg:irr} for details.
In this setting, we have finite set  $\rroutput_1, \dots, \rroutput_i, \ldots$ of outputs 
at each time step and a finite set 
$\rroutputvalid \subseteq \rroutput_1 \times \dots \times \rroutput_i \times \cdots$ 
of valid combinations of these outputs.  
The algorithm decides (randomly) at the very beginning of the run whether it is going to report truthfully 
(i.e., $\randomizedoutput = \perp$) or whether it is going to report some other output (i.e., $\randomizedoutput \in \rroutputvalid$).
In the latter case, the algorithm simply reports based on $\randomizedoutput$ regardless of the input. 
On the other hand, in the former case, the algorithm evaluates the query it receives and outputs truthfully in each step.

It is possible to see that event-level reports can be captured by the interactive mechanism
$\MER$. 
Indeed, let us choose $\rroutputvalid^{(x)}$ so
that $\rroutputvalid^{(x)}_{i}$ 
is the set of all combinations of possible 
reports that could be sent at $i$th second.
Then $q_x$ on $i$th iteration is the function that
creates the reports that would be sent on $i$th second.
(Note that we allow $q_x$ to depend on all the events, 
not just the one created during this second, this makes 
our privacy guarantee stronger than actually necessary.)

The main result of this section is that this is indeed $\eps$-DP.

\begin{theorem}
    \label{theorem:event-level-privacy}
    $\MER$ satisfies $\eps$-DP.
\end{theorem}

To prove this theorem, one may notice that while our definition of DP for interactive mechanisms
assumed that the adversary is deterministic it is not strictly necessary.
This follows from a simple {\em joint-convexity} property of DP.

\begin{fact}[Joint Convexity (see e.g. Lemma B.1 in \cite{chua24dpsgd})]\label{fact:joint-convexity}
Given two families of distributions $\{P_i\}_{i}$ and $\{Q_i\}_{i}$, if $P_i \approx_{\eps, \delta} Q_i$ for all $i$, then for all mixture distributions $P = \sum_i \alpha_i P_i$ and $Q = \sum_i \alpha_i Q_i$, it holds that $P \approx_{\eps, \delta} Q$.
\end{fact}

We extend the notion of a transcript to support distributions of adversaries (\Cref{alg:general-transcript-supporting-randmization}).
For any distribution $\adversaryDistribution$ over interactive adversaries, by applying the above fact for $P_{\adversary} = \transcript{\mechanism}{\adversary}$ and $Q_{\adversary} = \transcript{\mechanism^{-x}}{\adversary}$ for each $\adversary$ in the support of $\adversaryDistribution$, we get the following corollary.
\begin{corollary}
For all mechanisms $\mechanism$, if for all $x \in \cX$ and all interactive adversaries it holds that $\transcript{\mechanism}{\adversary} \approx_{\eps, \delta} \transcript{\mechanism^{-x}}{\adversary}$, then $\transcript{\mechanism}{\adversaryDistribution} \approx_{\eps, \delta} \transcript{\mechanism^{-x}}{\adversaryDistribution}$ holds for all distributions $\adversaryDistribution$ over interactive adversaries.
\end{corollary}

\begin{algorithm}[t]
    \caption{Interactive Transcript $\transcript{\mechanism}{\adversaryDistribution}$.}
    \label{alg:general-transcript-supporting-randmization}
    \begin{algorithmic}
        \STATE \textbf{Inputs:} $\triangleright$ Interactive mechanism $\mechanism$ with initial state $\mechState_0$,
        \STATE \phantom{\textbf{Inputs:}} $\triangleright$ A distribution of interactive adversaries $\adversaryDistribution$.
        \STATE Sample $\adversary\sim\adversaryDistribution$.
        \RETURN $\transcript{\mechanism}{\adversary}$.
    \end{algorithmic}
\end{algorithm}

\begin{algorithm}[t]
\caption{\IRR $\mathcal{I}_{\varepsilon, \rroutputvalid}$.}
\label{alg:irr}
\begin{algorithmic}
\STATE \textbf{Params:} 
    $\triangleright$ Privacy parameter $\varepsilon \in \R$,
\STATE \phantom{\textbf{Params:}} $\triangleright$ 
    Output set $\rroutputvalid \subseteq \rroutput_1 \times \dots \times \rroutput_i \times \cdots$.
\STATE \textbf{Inputs:} $\triangleright$ 
    State $S$ encoding $i \in \N$ and $\randomizedoutput \in \rroutputvalid \cup \{\perp\}$,
\STATE  \phantom{\textbf{Input:}} $\triangleright$ 
    Set of events $Y \subseteq \cY$ of events,
\STATE  \phantom{\textbf{Input:}} $\triangleright$ 
    Query $q: \cP(\cY) \to \rroutput_i$.
    
\IF{$i = 1$}
    \STATE Set $\randomizedoutput$ to $\perp$ with probability 
        $\frac{e^{\eps} - 1}{e^{\eps} + |\rroutputvalid|- 1}$, 
        otherwise set it to a random sample from $\rroutputvalid$
\ENDIF

\IF{$\randomizedoutput = \perp$}
	\RETURN $(\randomizedoutput, q(Y))$
\ELSE
    \RETURN $(\randomizedoutput, \randomizedoutput_i)$
\ENDIF
\end{algorithmic}
\end{algorithm}

\begin{algorithm}[t]
    \caption{Interactive mechanism 
    $\MER : \StateSet \times \cD \times \QuerySet \to \StateSet \times \Delta(\ResponseSet)$.}
    \label{alg:event-reports-interactive-mechanism}
    \begin{algorithmic}
        \STATE \textbf{Params:} $\triangleright$ Privacy parameter $\eps > 0$,
        \STATE \phantom{\textbf{Params:}} $\triangleright$ 
            Output sets $\{ \rroutputvalid^{(x)} \}_{x \in \cX}$.
        \STATE \textbf{Inputs:} $\triangleright$ State $\mechState$ encoding 
        $i \in \N$ and $\{ \randomizedoutput_x \in \rroutputvalid^{(x)} \cup \{\bot\} \}_{x \in \cX}$,
        \STATE \phantom{\textbf{Inputs:}} $\triangleright$ Database $D \in \cD$,
        \STATE \phantom{\textbf{Inputs:}} $\triangleright$ Queries
            $\{ q_x : \cY^* \to \rroutputvalid^{(x)}_i \}_{x \in \cX}$.

        \FOR{$x \in \cX$}
            \STATE $(s'_x, r_x) \gets \mathcal{I}_{\varepsilon, \rroutputvalid^{(x)}}(
                    i, \randomizedoutput_x, N_D(x), q_x
                )$
        \ENDFOR
        \RETURN $((i, \{s'_x\}_{x \in \cX}),  \{r_x\}_{x \in \cX})$
    \end{algorithmic}
\end{algorithm}

\begin{proof}[Proof of \Cref{theorem:event-level-privacy}]
    The proof consists of two parts: first, we show that $\MER$ is private when $\cX$ has only one element $x$;
    next, we prove that general privacy guarantee follows from this.
    
    Assume that $\cX = \{x\}$.
    Note that state of the mechanism does not change over the course of execution. 
    Let us denote the random variable for this state as $\randomizedoutput$. It is easy to see that for any 
    $\tilde{o} \in \rroutputvalid$,
    \begin{multline*}
        \Pr[\transcript{\MER}{\adversary} = \tilde{o} ~\mid~ \randomizedoutput = \tilde{o}] = \\
            \Pr[\transcript{\MER^{-x}}{\adversary} = \tilde{o} ~\mid~ \randomizedoutput = \tilde{o}] = 1.
    \end{multline*}
    Therefore, 
    \begin{align*}
        \Pr[\transcript{\MER}{\adversary} = \tilde{o}] &\ge \Pr[\randomizedoutput = \tilde{o}] = \frac{1}{e^{\eps} + |\rroutputvalid^{(x)}| - 1} \text{ and} \\
        \Pr[\transcript{\MER^{-x}}{\adversary} = \tilde{o}] &\ge \Pr[\randomizedoutput = \tilde{o}] = \frac{1}{e^{\eps} + |\rroutputvalid^{(x)}| - 1}.
    \end{align*}
    Furthermore, for any $\tilde{o} \in \rroutputvalid$,
    \begin{align*}
       & \Pr[\transcript{\MER}{\adversary} = \tilde{o}] \\
       &  = 
            \Pr[\transcript{\MER}{\adversary} = \tilde{o} \land \randomizedoutput = \perp] + 
            \Pr[\transcript{\MER}{\adversary} = \tilde{o} \land \randomizedoutput \neq \perp] \\
            & \le 
            \Pr[\randomizedoutput = \perp] + \Pr[\randomizedoutput = \tilde{o}] \\
            & = 
            \frac{e^{\eps} - 1}{e^{\eps} + |\rroutputvalid^{(x)}| - 1} + 
            \frac{1}{e^{\eps} + |\rroutputvalid^{(x)}| - 1} = 
            \frac{e^{\eps}}{e^{\eps} + |\rroutputvalid^{(x)}| - 1},
    \end{align*}
    which implies that $\transcript{\MER}{\adversary} \approx_{\eps, 0} \transcript{\MER^{-x}}{\adversary}$.
    
    First, we denote the mechanism $\MER$ operating on a dataset $\cX = \{x\}$ as $\MER^{(x)}$.
    Let us now assume that $\cX > 1$.
    We claim that for each $x \in \cX$, any adversary $\adversary$,
    and any transcript $\Pi$,
    there is a distribution $\adversaryDistribution'$
    over adversaries for 
    $\MER^{(x)}$ and a transcript $\Pi'$ for $\MER^{(x)}$ such that 
    \begin{align*}
        \Pr[\transcript{\MER}{\adversary} = \Pi] & =  \Pr[\transcript{\MER^{(x)}}{\adversaryDistribution'} = \Pi'], \\
        \Pr[\transcript{\MER^{-x}}{\adversary} = \Pi] & =
        \Pr[\transcript{\left(\MER^{(x)}\right)^{-x}}{\adversaryDistribution'} = \Pi'].
    \end{align*}
    Note that this implies that $\MER$ is $\eps$-DP.
    
    Let us now construct $\adversaryDistribution'$. First, we define an adversary $\adversary'_{\{ \randomizedoutput_{x'} \}_{x' \neq x}}$
    that runs $\adversary$ using actual response for $x$ and simulated responses for $x' \neq x$ using $\randomizedoutput_{x'}$
    (note that if the state is fixed, the mechanism is deterministic).
    It is clear that a distribution $\adversaryDistribution'$ that samples first $\{ \randomizedoutput_{x'} \}_{x' \neq x}$ and
    returns the adversary $\adversary'_{ \{\randomizedoutput_{x'} \}_{x' \neq x}}$ satisfies the desired condition.
\end{proof}

\section{Conclusion}\label{sec:discussion}

In this work, we modeled the summary reports in the Attribution Reporting API (ARA) and the Private Aggregation API (PAA), as well as the event-level reports in ARA.
We established formal DP guarantees for these mechanisms, even against the stringent notion of interactive adversaries that can influence the database in subsequent rounds based on responses in previous rounds.

\begin{acks}
    We thank the anonymous reviewers, for their feedback, 
    which significantly improved the clarity of this paper.
    We thank Roxana Geambasu, Pierre Tholoniat for discussions about \cite{tholoniat2024alistair}. 
    We are also grateful to Jolyn Yao and Christina Ilvento for their invaluable contributions, without which this paper might not have been possible.
\end{acks}

\newpage 
\bibliographystyle{ACM-Reference-Format}
\bibliography{main.bbl}

\newpage 

\appendix

\begin{table*}[t]
    \centering
    \begin{tabular}{l|p{13cm}}
        \toprule
        \textbf{Keyword} & \textbf{Meaning} \\
        \midrule
        {\em \Ad} & Advertisement shown on a publisher website. \\
        {\em ad-tech} & The entity that helps advertisers \& publishers with placement and measurement of digital \Ad{}s.\\
        {\em advertiser} & The entity that is paying for the advertisement, e.g. an online shoes shop. \\
        {\em aggregatable report} & An encrypted report that is sent to ad-tech every time a trigger is registered. \\
        {\em ARA} & Attribution Reporting API, that supports generation of {\em summary reports} and {\em event level reports} for attributed conversions. \\
        {\em attribution} & A conversion is attributed to an impression if the ads system beleives that this conversion happened due to this impression. \\
        {\em conversion} & An action on the advertiser website; for example, it could be a purchase. \\
        {\em impression} & An event where a user is exposed to some marketing information; for example, an ad is shown to the user. \\
        {\em key discovery} & The functionality that allows ad-techs to get a summary report without passing a list of keys of interest. \\
        {\em PAA} & Private Aggregation API, that supports generation of {\em summary reports}, corresponding to cross-website events. \\
        {\em publisher} & The entity that hosts the website that displays an advertisement, e.g. a news website.\\
        {\em requerying} & The functionality that allows ad-techs to process the same report multiple times using aggregation service. \\
        {\em shared storage} & The API that allows persisting a cross-web key-storage with read access being restricted to preserve privacy. \\
        {\em source} & The event on publisher website registered by the ad-tech with the browser; in typical use-cases, it corresponds to an impression. The $\sourceKey$ gets used in the generation of the aggregatable report for any trigger that get attributed to this source.\\
        {\em summary report} & The report obtained as a result of aggregating aggregatable reports and adding noise to the result. \\
        {\em trigger} & The event on advertiser website registered by the ad-tech that makes the ARA Client generate an aggregatable report; in typical use-cases it corresponds to conversions.\\
        \bottomrule
    \end{tabular}
    \caption{Glossary of commonly used terminology regarding the Privacy Sandbox.}
    \label{tab:glossary}
\end{table*}

\section{Analysis of (Truncated) Discrete Laplace Mechanism}\label{apx:tdlaplace}

We provide a proof of \Cref{fact:dlaplace-dp} for completeness.
The first component of this proof is the following tail-bound for discrete Laplace distributions.
\begin{lemma}
\label{lemma:truncated-discrete-laplace-tail}
    Let $\tau$ be an integer such that $\tau \ge \log(1 / \delta) / a + \Delta$. 
    Then $\Pr_{X \sim \DLap_{\tau}(a)}[X > \tau - \Delta] \le \delta$.
\end{lemma}
\begin{proof}
    \begin{align*}
        \Pr_{X \sim \DLap_{\tau}(a)}[X > \tau] &~=~ \frac{\sum_{x=\tau - \Delta + 1}^\tau e^{-a x}}{\sum_{x=-\tau}^\tau e^{-a |x|}}\\
        &~\le~ \frac{\sum_{x=\tau - \Delta + 1}^\tau e^{-a x}}{\sum_{x=0}^\tau e^{-a x}}\\
        &~\le~ \frac{\sum_{x=\tau - \Delta + 1}^\infty e^{-a x}}{\sum_{x=0}^\infty e^{-a x}}\\
        &~\le~ e^{-a(\tau - \Delta)} ~\le~ \delta.
    \end{align*}
    where the last inequality uses that $\tau \ge \log(1 / \delta) / a + \Delta$.
\end{proof}

\begin{proof}[Proof of \Cref{fact:dlaplace-dp}]
    First, we observe that by shifting the vectors $u$ and $v$, we can assume without loss of generality that $u = 0$.
    
    First, we consider the case of (untruncated) discrete Laplace noise. Let $P$ and $Q$ be the distributions of $u + \zeta$ and $v+ \zeta$ respectively for $\zeta \sim \DLap(a)^{\otimes d}$. In other words, for any $w \in \Z^d$, it holds that
    \[
    P(w) = \frac{1}{Z} e^{-a\|u-w\|_1}
    \qquad \text{and} \qquad
    Q(w) = \frac{1}{Z} e^{-a\|v-w\|_1}
    \]
    where $Z = (\frac{e^a + 1}{e^a - 1})^d$. Thus, we have
    \[
    \frac{P(w)}{Q(w)} = e^{-a (\|u-w\|_1 - \|v - w\|_1)}
    \]
    and it is thus easy to see that
    \[
    e^{-a\Delta} \le e^{-a\|u - v\|_1} \le \frac{P(w)}{Q(w)} \le e^{a \|u - v\|_1} \le e^{a\Delta}
    \]
    and thus, $P \approx_{\eps, 0} Q$ when $a = \eps / \Delta$.
    
    Next, moving to the case of truncated discrete Laplace noise, let $P$ and $Q$ be the distributions of $u + \zeta$ and $v+ \zeta$ respectively for $\zeta \sim \DLap_{\tau}(a)^{\otimes d}$. In particular, we have
    \[
        P(w) = \begin{cases}
        \frac{1}{Z} e^{-a \|u - w\|_1} & \text{if } \|u - w\|_\infty \le \tau\\
        0 & \text{if } \|u - w\|_\infty > \tau
        \end{cases}
    \]
    and similarly for $Q$, where $Z = \sum_{x=-\tau}^\tau e^{-a|x|}$. Let $S := \{ w : \|u - w\|_\infty \le \tau \text{ and } \|v - w\|_\infty \le \tau \}$. Similar to the case of untruncated case above, it follows that for all $w \in S$, it holds that $e^{-a\Delta} \le P(w)/Q(w) \le e^{a\Delta}$. Thus, for $a = \eps/\Delta$, it holds for all $E \subseteq \Z^d$ that
    \begin{align*}
        P(E) &~=~ P(E \cap S) + P(E \smallsetminus S)\\
        &~\le~ e^{\eps} Q(E \cap S) + P(E \smallsetminus S)\\
        &~\le~ e^{\eps} Q(E) + P(\Z^d \smallsetminus S)
    \end{align*}
    Thus, $P \approx_{\eps, \delta} Q$ where $\delta := P(\Z^d \smallsetminus S)$. To complete the proof, we need to show that when $\tau \ge \Delta (1 + \log(s / \delta) / \eps)$, it holds that $P(\Z^d \smallsetminus S) \le \delta$; recall that $u$ and $v$ differ on $s$ coordinates. We have
    \begin{align}
    P(\Z^d \smallsetminus S)
    &~=~ 1 - P(S)\nonumber\\
    &~=~ 1 - \prod_{i=1}^d \Pr_{w_i \sim u_i + \DLap_{\tau}(a)} [|w_i - v_i| \le \tau]\nonumber\\
    &\textstyle~\le~ s \cdot \Pr_{X \sim \DLap_{\tau}(a)}[X > \tau - \Delta]\label{eq:dlaplace-union-bound}
    \end{align}
    where, we use that when $u_i = v_i$, we have
    $$\Pr_{w_i \sim u_i + \DLap_{\tau}(a)} [|w_i - v_i| \le \tau] = 1$$
    and when $u_i \ne v_i$, we have
    \begin{align*}
    \Pr_{w_i \sim u_i + \DLap_{\tau}(a)} [|w_i - v_i| \le \tau]
    &~=~ \Pr_{X \sim \DLap_{\tau}(a)} [|X + u_i - v_i| \le \tau]\\
    &~\ge~ 1 - \Pr_{X \sim \DLap_{\tau}(a)}[X > \tau - \Delta].
    \end{align*}
    Note that \Cref{lemma:truncated-discrete-laplace-tail} implies that
    \[
        \Pr_{X \sim \DLap_{\tau}(a)}[X > \tau - \Delta] \le \delta / s.
    \]
    since $\tau \ge \Delta(1 + \log(s / \delta) / \eps)$. 
    Combining with \Cref{eq:dlaplace-union-bound}, we get that $P(\Z^d \smallsetminus S) \le \delta$, thereby completing the proof.
\end{proof}

\section{Extension to other notions of DP}\label{sec:AzCDP}

While we primarily studied $(\eps, \delta)$-DP notion in this paper, our definitions and techniques readily extend to any other notion of DP that admits privacy filters. In particular, we could consider the following notion of Approximate zero Concentrated DP.

\begin{definition}[\cite{bun16concentrated}]\label{def:approx-zCDP}
    Two distributions $P$, $Q$ are said to be 
    \emph{$(\rho, \delta)$-AzCDP-indistinguishable}\footnote{%
        $(\rho, \delta)$-AzCDP is referred to as \emph{$\delta$-approximate $\rho$-zCDP} in \cite{bun16concentrated}.
    }
    if there exist events $W$ and $W'$ such that
    \begin{gather*}
        P(W) \ge 1 - \delta,\quad Q(W') \ge 1 - \delta, \\
        \divergence{\alpha}\left(P|_W ~\|~ Q|_{W'}\right) \le \rho \alpha, \quad \text{ and } \quad
        \divergence{\alpha}\left(Q|_{W'} ~\|~ P|_W\right) \le \rho \alpha,
    \end{gather*}
    where for any $\alpha > 1$, 
    $\divergence{\alpha}(U ~\|~ V) := \frac{1}{\alpha - 1} \log\left(\int U(x)^\alpha V(x)^{1 - \alpha} \mathop{dx}\right)$ 
    denotes the $\alpha$-R\'{e}nyi divergence between $U$ and $V$.
\end{definition}

All the proof techniques we applied for $(\eps, \delta)$-DP also extend to hold for $(\rho, \delta)$-AzCDP. In particular, we have to rely on the privacy filter $\phi_{\rho, \delta}$ for AzCDP that is defined similarly.
\[
    \filter_{\rho, \delta}((\rho_1, \delta_1), \dots, (\rho_n, \delta_n)) := \begin{cases}
        \thumbsup & \text{if } \sum\limits_{i = 1}^n \rho_i \le \rho \text{ \& } \sum\limits_{i = 1}^n \delta_i \le \delta\\
        \thumbsdown & \text{otherwise}.
    \end{cases}
\]

\begin{lemma}[\cite{whitehouse23fully}]\label{lem:filters_2}
    For all $\rho \ge 0$ and $\delta \in [0, 1]$, the universal interactive mechanism $\mechanism_{\filter_{\rho, \delta}}$ satisfies $(\rho, \delta)$-AzCDP.
\end{lemma}

AzCDP is useful in performing privacy accounting of the Gaussian mechanism, where using standard $(\eps, \delta)$-DP notion results in sub-optimal privacy guarantees under composition.%
	
\end{document}